\documentclass[letterpaper, 10 pt, conference]{ieeeconf}

\IEEEoverridecommandlockouts 
\makeatletter

\let\proof\@undefined
\let\endproof\@undefined
\makeatother
\let\labelindent\relax

\usepackage{jhoagg}

\usepackage{mathtools}
\usepackage{amsfonts}
\usepackage{amssymb,latexsym}
\usepackage[psamsfonts]{eucal}
\usepackage{amsthm}
\usepackage{graphicx}
\usepackage[inline]{enumitem}
\usepackage{indentfirst}
\usepackage{setspace}
\usepackage{microtype}
\usepackage{threeparttable}
\usepackage{float}
\usepackage{cleveref}
\usepackage{afterpage}
\usepackage{placeins}
\usepackage{cancel}
\usepackage{cite}
\usepackage{leftidx}
\usepackage{breqn}
\usepackage[export]{adjustbox}
\usepackage{comment}
\usepackage[ruled, linesnumbered, lined]{algorithm2e}
\usepackage[dvipsnames]{xcolor}
\usepackage{xcolor}
\usepackage{soul}
\usepackage{siunitx}
\usepackage{subcaption}

\usepackage{tikz}
\usepackage{pgf}
\usepackage{pgfplots}
\usepgflibrary{plothandlers}
\usepgfplotslibrary{external} 
\pgfkeys{/pgf/number format/.cd,1000 sep={}}  
\usetikzlibrary{calc,shapes.misc,plotmarks}
\pgfplotsset{compat=1.13}

\let\originalleft\left
\let\originalright\right
\renewcommand{\left}{\mathopen{}\mathclose\bgroup\originalleft}
\renewcommand{\right}{\aftergroup\egroup\originalright}

\newcounter{thm} 
\newtheorem{theorem}[thm]{\indent Theorem}

\newtheorem{assumption}{\indent Assumption}

\newtheorem{proposition}{\indent Proposition}

\newtheorem{lemma}{\indent Lemma}

\newtheorem{corollary}{\indent Corollary}

\newtheorem{definition}{\indent Definition}

\newtheorem{remark}{\indent Remark}

\newtheorem{example}{\indent Example}

\newtheorem{Simulation}{Simulation}

\newtheorem{fact}{\indent Fact}

\newtheorem{conjecture}{\indent Conjecture}

\newtheorem{experiment}{\indent Experiment}

\allowdisplaybreaks

\renewcommand{\theenumi}{{\it (\alph{enumi})}}

\usepackage{cite}
\usepackage{accents}

\newlength\figureheight 
\newlength\figurewidth

\allowdisplaybreaks

\graphicspath{ {Figures/} }
\DeclareGraphicsExtensions{.png}

\DeclareMathAlphabet{\mathcal}{OMS}{cmsy}{m}{n} 

\crefname{equation}{}{}

\usepackage[ruled, linesnumbered, lined]{algorithm2e}
\SetKwInput{KwInit}{Initialize}
\SetKwFunction{SafeSet}{SafeSet}
\SetKwFunction{GetSafe}{GetSafe}
\SetKwFunction{GetUnsafe}{GetUnsafe}

\SetCommentSty{mycommfont}

\begin{document}
	\title{Adaptive Control Barrier Functions with Vanishing Conservativeness Under Persistency of Excitation}
	
	\author{Ricardo Gutierrez and Jesse B. Hoagg
		\thanks{R. Gutierrez and J. B. Hoagg are with the Department of Mechanical and Aerospace Engineering, University of Kentucky, Lexington, KY, USA. (e-mail: Ricardo.Gutierrez@uky.edu, jesse.hoagg@uky.edu).}
		\thanks{R. Gutierrez is supported by the Fulbright-SENACYT Scholarship. This work is also supported in part by the National Science Foundation (1849213) and Air Force Office of Scientific Research (FA9550-20-1-0028).}
	}
	\maketitle
	
	\begin{abstract}
This article presents a closed-form adaptive control-barrier-function (CBF) approach for satisfying state constraints in systems with parametric uncertainty.
This approach uses a sampled-data recursive-least-squares algorithm to estimate the unknown model parameters and construct a nonincreasing upper bound on the norm of the estimation error. 
Together, this estimate and upper bound are used to construct a CBF-based constraint that has nonincreasing conservativeness.
Furthermore, if a persistency of excitation condition is satisfied, then the CBF-based constraint has vanishing conservativeness in the sense that the CBF-based constraint converges to the ideal constraint corresponding to the case where the uncertainty is known. 
In addition, the approach incorporates a monotonically improving estimate of the unknown model parameters---thus, this estimate can be effectively incorporated into a desired control law. 
We demonstrate constraint satisfaction and performance using 2 two numerical examples, namely, a nonlinear pendulum and a nonholonomic robot. 
\end{abstract}
	
	
	

\section{Introduction}

Control systems are often required to respect state constraints (e.g., safety constraints) and achieve performance requirements such as minimizing a performance-based cost function or achieving tasks such as formation control \cite{borrmann2015control,lippay2021formation,heintz2023formation}, locomotion \cite{nguyen2015safety}, or destination seeking 
%
%
%
%
%
%
%
%
performance subject to state constraints can be addressed using a variety of methods including model predictive control (e.g., \cite{borrelli2017predictive,bemporad2002model,tondel2003algorithm,zeng2021safety} and barrier function approaches (e.g., \cite{prajna2007framework, wieland2007constructive,ames2016control,nguyen2016exponential,jin2018adaptive,xiao2021high,tan2021high}).
%
%
%
%
%
%
%
%
%
%

Control barrier functions (CBFs) are an approach to achieve forward invariance of a set that satisfies state constraints (e.g., safety) \cite{wieland2007constructive}. 
CBFs are commonly implemented as constraints in real-time optimization control methods (e.g., quadratic programs) in order to guarantee state-constraint satisfaction while also attempting to minimize a performance based cost \cite{ames2016control}. 
The minimum-intervention problem is one common example, where a desired control $u_\rmd$ is designed to satisfy performance requirements but may not account for state constraints. 
Then, a control is generated that is as close as possible to $u_\rmd$ while satisfying a CBF constraint that guarantees state-constraint satisfaction. 
Examples of CBF applications include adaptive cruise control \cite{ames2016control,xiao2022event}, satellites \cite{breeden2023robust}, unmanned air vehicles \cite{wang2023multi, zheng2023constrained}, and ground robots \cite{ames2019control,cosner2021measurement,rabiee2023automatica, rabiee2024closed, safari2024time}.
%
%

CBF methods are model based, and model uncertainty can lead to both state-constraint violations and poor performance. 
Robust CBF methods such as \cite{jankovic2018robust,cohen2022robust} address model uncertainty with respect to state-constraint satisfaction. 
However, robust CBF methods adopt worst-case bounds, which can lead to overly conservative system behavior. 
Moreover, robust CBFs do not specifically address performance degradation that typically arises due to model uncertainty. 
In other words, robust CFBs do not directly address that a desired control $u_\rmd$ may lead to poor performance because of model uncertainty.

As an alternative, adaptive CBF techniques can also address model uncertainty (e.g., \cite{taylor2020adaptive,lopez2020robust,nguyen20221,cohen2022high,zeng2023robust}). 
%
%
%
%
%
For example, \cite{taylor2020adaptive} presents an adaptive CBF for state-constraint satisfaction in the presence of model uncertainty. 
However, the approach in \cite{taylor2020adaptive} can be conservative in the sense that the CBF constraint imposed for safety may be more conservative than needed. 
In \cite{lopez2020robust}, set-membership identification is combined with the basic approach of \cite{taylor2020adaptive} to reduce conservativeness. 
Higher-order adaptive CBF methods are presented in \cite{cohen2022high}.

This article presents an adaptive CBF approach that has 2 key features. 
First, the adaptive CBF for the state-constraint satisfaction is guaranteed to have nonincreasing conservativeness, that is, the conservativeness can be captured by a decreasing function of time. 
Moreover, the conservativeness vanishes in the case of persistency of excitation. 
Second, the approach incorporates an estimate of the unknown model parameters, where norm of the estimation error is guaranteed to be nonincreasing and vanishes in the case of persistency of excitation. 
Thus, the estimate of the unknown model parameters can be effectively incorporated into a desired control law. 
Together, these 2 features can help improve performance while ensuring state-constraint satisfaction.

\section{Notation}
Let $q:\mathbb{R}^{n} \rightarrow \mathbb{R}$ be continuously differentiable. 
Then, $q^\prime:\mathbb{R}^{n} \rightarrow \mathbb{R}^{1\times n}$ is defined by $q^\prime(x) \triangleq \textstyle \frac{\partial q(x)}{\partial x}$. 
The Lie derivative of $q$ along the vector fields of $\upsilon:\mathbb{R}^{n} \rightarrow \mathbb{R}^{n \times l}$ is defined as $L_\upsilon q(x) \triangleq q^\prime(x) \upsilon(x)$.
If $l=1$, then for all positive integers $r$,  define $L_\upsilon^{r}q(x) \triangleq L_{\upsilon} L_{\upsilon}^{r-1} q(x)$.
In this paper, we assume that all functions are sufficiently smooth such that all derivatives that we write exist and are continuous.
	
The boundary of $B \subseteq \BBR^n$ is denoted by $\mbox{bd } B$. 
Let $\|\cdot\|$ denote the 2 norm on $\BBR^n$.
The maximum eigenvalue of $A \in \BBR^{n \times n}$ is denoted by $\lambda_{\rm{max}}(A)$.

\section{Problem Formulation}\label{sec:problem formulation}

Consider the dynamic system
\begin{equation}
		\Dot{x}(t)=f(x(t))+\phi(x(t))\theta_{*}+g(x(t))u(t),
		\label{eq:dyn}
\end{equation}
where $x(t) \in \mathbb{R}^n$ is the state;
$x(0)=x_{0} \in \BBR^n$ is the initial condition;
$u: [0, \infty) \rightarrow \mathbb{R}^m$ is the control; $f: \mathbb{R}^n \rightarrow \mathbb{R}^{n}$, $g: \mathbb{R}^n \rightarrow \mathbb{R}^{n \times m}$, and $\phi: \mathbb{R}^{n} \rightarrow \mathbb{R}^{n \times p}$ are locally Lipshchitz continuous on $\mathbb{R}^{n}$;
$\theta_{*} \in \Theta$ is an unknown parameter; 
and $\Theta \subset \mathbb{R}^p$ is assumed to be bounded and known.

	
Let $\psi_{0}:\mathbb{R}^{n} \rightarrow \mathbb{R}$ be continuously differentiable, and define the \textit{safe set}
    \begin{equation}
	C_{0} \triangleq \{x \mid \psi_{0}(x)\geq0\},
	\label{eq:psi0}
    \end{equation}
which is the set of states that satisfy the state constraint. 
We make the following assumption:

\begin{assumption}{\rm 
There exists a positive integer $d$ such that for all $i\in\{0,1,\ldots,d-2\}$ and all $x \in \BBR^n$, $L_{g}L_{f}^{i}\psi_{0}(x)=0$ and               $L_{\phi}L_{f}^{i}\psi_{0}(x)=0$.
}\label{assum:input}
\end{assumption}

Assumption~\ref{assum:input} implies that the relative degree of $\psi_0$ with respect to the control $u$ and the unknown parameter $\theta_*$ is at least $d$ on $\BBR^n$. 
Thus, we use a higher-order approach to construct a candidate CBF. 
For all $i \in \{1,\ldots,d-1\}$, let $\alpha_{i-1}: \mathbb{R} \rightarrow \mathbb{R}$ is a $(d-i)$-times continuously differentiable extended class-$\SK$ function, and define 
    \begin{equation}
	\psi_{i}(x) \triangleq L_{f}\psi_{i-1}(x)+\alpha_{i-1} \big( \psi_{i-1}(x) \big).
	\label{eq:psii}
    \end{equation}
For all $i \in \{1,\ldots,d-1\}$, the zero-superlevel set of $\psi_i$ is given by 
    \begin{equation}
        C_{i} \triangleq \{x \mid \psi_{i}(x)\geq0\}.
        \label{eq:Ci}
    \end{equation}    
We make the following assumption:

\begin{assumption}
\rm{For all $x \in \mathrm{bd} \hspace{1mm} C_{d-1}$, $L_{g}L_{f}^{d-1}\psi_{0}(x) \neq 0$.} 
	\label{assum:assump2}
\end{assumption}
Assumption~\ref{assum:assump2} implies that $\psi_0$ has relative degree $d$ with respect to the control $u$ on $\mbox{bd } C_{d-1}$. 
Thus, 
\begin{equation}
\bar{C} \triangleq \bigcap_{i=0}^{d-1} C_{i}
        \label{eq:safe_set_mixed}
\end{equation}
is control forward invariant.

Next, consider the desired control $u_{\rm{d}}:\mathbb{R}^{n} \times \mathbb{R}^{p} \rightarrow \mathbb{R}^{m}$, which is designed to satisfy performance requirements but may not satisfy the state constraint that  $x(t) \in C_{0}$ for all $t$.
The desired control $u_{\rm{d}}(x,\hat \theta)$ is a function of the state $x$ and can be parameterized $\hat \theta$, which represents an estimate of the unknown parameter $\theta_*$. 
In other words, $u_{\rm{d}}(x,\theta_*)$ is the ideal desired control.

Consider the cost function $J:\mathbb{R}^{n} \times \mathbb{R}^{p} \times \mathbb{R}^{m} \rightarrow \mathbb{R}$ defined by
    \begin{equation}
        J(x,\hat{\theta},\hat{u})\triangleq \frac{1}{2} \Big ( \hat{u}-u_{\rm{d}}(x,\hat{\theta}) \Big )^{\rmT} H ( x,\hat{\theta}) \Big (\hat{u}-u_{\rm{d}}(x,\hat{\theta}) \Big ), 
	\label{eq:cost_func0}
    \end{equation}
where  $H : \mathbb{R}^{n} \times \mathbb{R}^{p} \rightarrow \BBR^{m \times m}$ is locally Lipschitz, and for all $(x,\hat{\theta}) \in \mathbb{R}^{n} \times \mathbb{R}^{p}$, $H(x,\hat{\theta})$ is positive definite.

The objective is twofold. 
First, design a feedback control that for all $(x,\hat{\theta})$, minimizes \eqref{eq:cost_func0} subject to the state constraint that $x(t) \in \bar C$ for all $t \ge 0$. 
Second, we aim to estimate $\theta_*$ so that the desired control $u_{\rm{d}}(x,\hat \theta)$ approaches the ideal desired control $u_{\rm{d}}(x,\theta_*)$.

\section{Parameter Estimation} \label{sec:Method}
This section provides a method of estimating $\theta_*$, and provides a computable upper bound on the norm of the parameter estimation error. 
This computable upper bound is nonincreasing and used in subsequent sections to design a safe control that is not overly conservative.

Unless otherwise stated, all expressions with subscript $k$ are for all $k \in \mathbb{N} \triangleq \{0,1,2,3,\ldots\}$.
Let $t_{0}=0$, and let $t_{k+1} > t_{k}$.
Define
	\begin{equation}
		\Phi_{k} \triangleq \int_{t_{k}}^{t_{k+1}} \phi(x(t)) \rmd t,
		\label{eq:phik}
	\end{equation}
and for all negative integers $i$, let $\Phi_{i} = 0_{n \times p}$. 
Next, let $k_{\rm{n}}$ be a nonnegative integer, and define
	\begin{equation}
		P_{k}\triangleq (\sigma_{k} I_{p}+\Omega_k )^{-1} \in \mathbb{R}^{p \times p},
		\label{eq:Pk}
	\end{equation}
where $\sigma_{k} > 0$ and
\begin{equation}
\Omega_k \triangleq \sum_{i=k-k_{\rm{n}}}^{k}\Phi_{i}^{\rm{T}}\Phi_{i}. \label{eq:Omega}
\end{equation}
Define     
\begin{align}
y_{k} &\triangleq x(t_{k+1})-x(t_{k}) \nn\\
&\qquad -\int_{t_{k}}^{t_{k+1}} f(x(t))+g(x(t))u(t) \rmd t.
		\label{eq:yk}
\end{align}
The next result shows that $y_k$ is a linear regression in $\Phi_k$ and the unknown parameter $\theta_*$. 

\begin{proposition}{\rm 
For all $k \in \mathbb{N}$, $y_{k} = \Phi_{k}\theta_{*}$.
}\label{prop:yk}
\end{proposition}

\begin{proof}[\indent Proof]
It follows from \eqref{eq:dyn} that
\begin{equation*}
\phi(x(t))\theta_{*}=\Dot{x}(t)-f(x(t))-g(x(t))u(t),
\end{equation*}
and integrating over $[t_{k},t_{k+1})$ yields
\begin{align*}
\Phi_{k}\theta_{*}&= x\big(t_{k+1}\big)-x\big(t_{k}\big)\\
    			     & \qquad -\int_{t_{k}}^{t_{k+1}} f(x(t))+g(x(t))u(t) \rmd t.
\end{align*}
Then, \eqref{eq:yk} implies that $y_{k} = \Phi_{k}\theta_{*}$.
\end{proof}

Next, consider the regularized finite-horizon least-squares cost $\mathcal{J}_{k}:\mathbb{R}^{p} \rightarrow \mathbb{R}$ defined by
\begin{equation*}
\mathcal{J}_{k}(\hat{\theta})\triangleq \sigma_{k}\|\hat{\theta}-\theta_{k}\|^{2}+ \sum_{i=k-k_{\rm{n}}}^{k} \| \Phi_{i}\hat{\theta}-y_{i}\|^{2}.
\end{equation*}
For all $k \in \mathbb{N}$, the minimizer of $\SJ_k$ is given by
	\begin{equation}
		\theta_{k+1} \triangleq P_{k}\bigg(\sum_{i=k-k_{\rm{n}}}^{k}\Phi_{i}^{\rm{T}}y_{i}+\sigma_{k}\theta_{k}\bigg),
		\label{eq:theta_k}
	\end{equation}
	where $\theta_{0} \in \mathbb{R}^{p}$. 
 The finite-horizon least-squares estimator \eqref{eq:theta_k} can also be expressed in recursive form. 
See \cite{ali2016stability} for details.
%

	Next, define
	\begin{equation}
		\tilde{\theta}_{k}\triangleq \theta_{k} - \theta_{*},
		\label{eq:theta_tild}
	\end{equation}
        and consider the Lyapunov-like function
	\begin{equation}
		V_{k}\triangleq \tilde{\theta}_{k}^{\rm{T}}\tilde{\theta}_{k},
		\label{eq:Vk}
	\end{equation} 
        and the Lyapunov-like difference
	\begin{equation}
		\Delta V_{k} \triangleq \tilde{\theta}_{k+1}^{\rm{T}}\tilde{\theta}_{k+1}-\tilde{\theta}_{k}^{\rm{T}}\tilde{\theta}_{k}.
		\label{eq:dVk}
	\end{equation}    
The next result provides properties of the finite-horizon least-squares estimator \Cref{eq:theta_k,eq:phik,eq:Omega,eq:Pk,eq:yk}.

\begin{proposition}\label{prop:prop}{\rm
The following statements holds:
\begin{enumerate}

\item \label{prop:prop:Vk} 
For all $k \in \mathbb{N}$, $\Delta V_{k}=-\tilde{\theta}_{k}^{\rm{T}}\big(I_{p}-\sigma_{k}^{2}P_{k}^{2}\big)\tilde{\theta}_{k} \leq 0$.

\item \label{prop:prop:tk} 
$\{\theta_{k}\}_{k=0}^{\infty}$ is bounded.

\item \label{prop:prop:tk2}
Assume there exists $k_{\rm{i}} \in \mathbb{N}$ such that for all $k \geq k_{i}$, $\Omega_k$ is full rank. 
Then, $\theta_{k} \rightarrow \theta_{*}$ as $k \rightarrow \infty$.
\end{enumerate}
}\end{proposition}

    
\begin{proof}[\indent Proof]
Note that it follows from \eqref{eq:theta_k} that
        \begin{equation} 
		\tilde{\theta}_{k+1}=\tilde{\theta}_{k}-P_{k} \Omega_k \tilde{\theta}_{k}.
		\label{eq:tildethetaplus}
	\end{equation} 
Since \eqref{eq:Pk} implies that $\Omega_k =P_{k}^{-1}-\sigma_{k} I_{p}$, it follows from \eqref{eq:tildethetaplus} that 
\begin{equation} 
		\tilde{\theta}_{k+1}=\sigma_{k} P_{k} \tilde{\theta}_{k}.
		\label{eq:tildethetaplus.2}
\end{equation} 
Evaluating the Lyapunov-like difference \eqref{eq:dVk} along \eqref{eq:tildethetaplus.2} yields
\begin{align}
		\Delta V_{k}&=-\tilde{\theta}_{k}^{\rm{T}}\big(I_{p}-\sigma_{k}^{2}P_{k}^{2}\big)\tilde{\theta}_{k}\nn\\
        &\le  \Big ( \sigma_{k}^{2} \lambda_{\rm{max}}(P_{k})^2 -1 \Big ) \| \tilde{\theta}_{k} \|^2,\label{eq:dVk.2}
\end{align}
and it follow from \eqref{eq:Pk} that $\sigma_{k}^{2} \lambda_{\rm{max}}(P_{k})^2 -1 \le 0$, which confirms \ref{prop:prop:Vk}.

To prove \ref{prop:prop:tk}, since $\Delta V_{k} \leq 0$, it follows that $\tilde{\theta}_{k}$ is bounded. 
Since, in addition, $\tilde{\theta}_{k} = \theta_{k}-\theta_{*}$, it follows that $\theta_{k}$ is bounded.

To prove \ref{prop:prop:tk2}, since $\Omega_k$ is full rank, it follows from \eqref{eq:Pk} that $\sigma_{k}^{2} \lambda_{\rm{max}}(P_{k})^2 -1 < 0$, which combined with \eqref{eq:dVk.2} implies that $\Delta V_{k} < 0$ for all $\tilde\theta_k \ne 0$. 
Thus, $\tilde\theta_k \to 0$ as $k \to \infty$, which confirms \ref{prop:prop:tk2}.
\end{proof}

    Define
    \begin{equation}
        \bar{y}_{i}(\theta_{k})\triangleq \Phi_{i}\theta_{k}-y_{i},
        \label{eq:yhat}
    \end{equation}
    and define
	\begin{equation}
            \begin{aligned}
    		\tau_{k} & \triangleq \frac{2}{\sigma_{k}}\bigg[-\sum_{i=k-k_{\rm{n}}}^{k}||\bar{y}_{i}(\theta_{k})||^{2}\\
    		& \qquad +\bigg(\sum_{i=k-k_{\rm{n}}}^{k}\bar{y}_{i}^{\rm{T}}(\theta_{k})\Phi_{i}\bigg)P_{k}\bigg(\sum_{i=k-                             
                k_{\rm{n}}}^{k}\Phi_{i}^{\rm{T}}\bar{y}_{i}(\theta_{k})\bigg)\bigg]\\
    		& \qquad +\bigg(\sum_{i=k-k_{\rm{n}}}^{k}\bar{y}_{i}^{\rm{T}}(\theta_{k})\Phi_{i}\bigg)P_{k}^{2}\bigg(\sum_{i=k-
                k_{\rm{n}}}^{k}\Phi_{i}^{\rm{T}}\bar{y}_{i}(\theta_{k})\bigg),
            \end{aligned}
		\label{eq:tauk}
	\end{equation}
 
%
%
%

 The next result shows that the computable quantity $\tau_k$ is equal to the Lyapunov-like difference $\Delta V_k$, which cannot be computed from its definition because $\theta_*$ is unknown.

\begin{lemma}{\rm
For all $k\in \mathbb{N}$, $\tau_{k}=\Delta V_{k}$.
		\label{lemma:tauk}
}\end{lemma}

\begin{proof}[\indent Proof]
\Cref{prop:yk} and \eqref{eq:yhat} imply $\bar{y}_{i}(\theta_{k})=\Phi_{i}\tilde{\theta}_{k}$, and substituting into \eqref{eq:tauk} yields 
\begin{equation*}
\tau_{k} = \frac{2}{\sigma_{k}}\tilde{\theta}_{k}^{\rm{T}} (-\Omega_k + \Omega_k P_k \Omega_k )\tilde{\theta}_{k}+ \tilde{\theta}_{k}^{\rm{T}}\Omega_k P_k^{2}\Omega_k \tilde{\theta}_{k}. 
\end{equation*}
Since \eqref{eq:Pk} implies that $\Omega_k =P_{k}^{-1}-\sigma_{k} I_{p}$, it follows that $\tau_{k} = \tilde{\theta}_{k}^{\rm{T}} (-I_{p}+\sigma_{k}^{2}P_{k}^{2} )\tilde{\theta}_{k}$, which combined with \ref{prop:prop:Vk} of \Cref{prop:prop} yields $\tau_{k}=\Delta V_{k}$.
\end{proof}

Next, let $\nu_{0} >0$ satisfy
	\begin{equation}
		\nu_{0} \geq \sup_{\bar{\theta} \hspace{0.25mm} \in \hspace{0.25mm}\Theta}\big|\big|\theta_{0}-\bar{\theta}\big|\big|,
		\label{eq:c0}    
	\end{equation}
and note that $\nu_0$ satisfying \eqref{eq:c0} can be computed because the set $\Theta$ is known. 
For all $k \in \mathbb{N}$, define
	\begin{equation}
		\nu_{k+1}\triangleq \min\bigg(\sigma_{k}\lambda_{\rm{max}}(P_{k})\nu_{k} \hspace{1mm}, \hspace{1mm} \sqrt{\nu_{k}^{2}+\tau_{k}}\bigg),
		\label{eq:ck1}
	\end{equation}
 and the following result shows that $\nu_k$ is a nonincreasing upper bound on the magnitude of the estimation error $\tilde \theta_k$.
	
\begin{proposition}          \label{prop:tkck2}\rm 
The following statements holds:
\begin{enumerate}

\item \label{prop:tkck2:ck1} 
For all $k \in \mathbb{N}$, $\nu_{k} \geq ||\tilde{\theta}_{k}||$.

\item \label{prop:tkck2:ck2} 
For all $k \in \mathbb{N}$, $\nu_{k+1} \leq \nu_{k}$.

\item \label{prop:tkck2:ck3} 
$\{\nu_{k}\}_{k=0}^{\infty}$ is bounded.

\item \label{prop:tkck2:ck4} 
   Assume there exists $k_{\rm{i}} \in \mathbb{N}$ such that for all $k \geq k_{i}$, $\Omega_k$ is full rank. 
   Then, $\nu_{k}$ $\rightarrow 0$ as $k \rightarrow \infty$.
\end{enumerate}
  
\end{proposition}

\begin{proof}[\indent Proof]

To prove \ref{prop:tkck2:ck1}, it follows from \eqref{eq:dVk} and \ref{prop:prop:Vk} of \Cref{prop:prop} that 
\begin{equation}
||\tilde{\theta}_{k+1}|| \leq \sigma_{k}\lambda_{\rm{max}}(P_{k})||\tilde{\theta}_{k}||.
            \label{eq:theta_tild1}
\end{equation}
We use induction on $k$. 
First, since $\theta_{*} \in \Theta$, it follows from \eqref{eq:c0} that $\nu_{0} \geq \|\tilde{\theta}_{0}\|$, which proves \ref{prop:tkck2:ck1} for $k=0$. 
Next, let $j \in \BBN$, and assume for induction that $\nu_{j} \geq \|\tilde{\theta}_{j}\|$. 
Thus, \eqref{eq:ck1} implies that
\begin{equation*}
    	\nu_{j+1} \geq  \min\Big( \sigma_{j} \lambda_{\rm{max}}\big(P_{j}\big)||\tilde{\theta}_{j}||,\sqrt{||\tilde{\theta}_{j}||^{2}+\tau_{j}}\Big),
\end{equation*}
which combined with \eqref{eq:theta_tild1} yields
\begin{equation}
\nu_{j+1} \geq \min \Big(\|\tilde{\theta}_{j+1}\| , 
\sqrt{\|\tilde{\theta}_{j}\|^{2}+\tau_{j}}\big). \label{eq:nu_j+1}
\end{equation}
Since \Cref{lemma:tauk} and \eqref{eq:dVk} imply that $\sqrt{\|\tilde{\theta}_{k}\|^{2} + \tau_{k}} = \|\tilde{\theta}_{k+1}\|$, it follows from \eqref{eq:nu_j+1} that $\nu_{j+1} \geq \|\tilde{\theta}_{j+1}\|$, which confirms \ref{prop:tkck2:ck1}.

To prove \ref{prop:tkck2:ck2}, it follows from \Cref{lemma:tauk} and \ref{prop:prop:Vk} of  \Cref{prop:prop} that $\sqrt{\nu_{k}^{2}+\tau_{k}} \leq \nu_{k}$, which combined with \eqref{eq:ck1} confirms \ref{prop:tkck2:ck2}.  

To prove \ref{prop:tkck2:ck3}, since \eqref{eq:ck1} implies that $\nu_k$ is nonnegative, it follows from \ref{prop:tkck2:ck2} that $\{\nu_{k}\}_{k=0}^{\infty}$ is bounded, which confirms \ref{prop:tkck2:ck3}.


	To prove \ref{prop:tkck2:ck4}, note that, since $\Omega_k$ is full rank, it           follows that $\lambda_{\rm{max}}\big(\sigma_{k}P_{k}\big) < 1$ and $\tau_{k} \leq 0$, which implies that $\lambda_{\rm{max}}\big(\sigma_{k}P_{k}\big)\nu_{k}      < \nu_{k}$ and $\sqrt{\nu_{k}^{2}+\tau_{k}} \leq \nu_{k}$ and consequently from \eqref{eq:ck1}, $\nu_{k+1} < \nu_{k}$. From \ref{prop:tkck2:ck1}, note that $||\tilde{\theta}_{k+1}|| \leq \nu_{k+1} < \nu_{k}$. Finally, since $||\tilde{\theta}_{k+1}|| \rightarrow 0$ as $k \rightarrow \infty$, it follows that $\nu_{k} \rightarrow 0$ as $k \rightarrow \infty$.
	\end{proof}

The next section uses the estimate $\theta_k$ and the computable upper bound $\nu_k$ to develop a safe an optimal control.

\section{Safe and Optimal Control}

Let $\xi: [0, \infty) \rightarrow [0,1]$ be a nondecreasing and continuously differentiable function such that for all $t \in (-\infty,0]$,  $\xi(t)=0$ and for all $t \in [1,\infty)$, $\xi(t)=1$. 
The following example provides one possible choice for $\xi$.

\begin{example}\label{ex:xi}\rm
Let $\displaystyle \eta \geq 1$, and consider
		\begin{equation}
			\xi(t) \triangleq
			\begin{cases} 
				0, & t \in (-\infty,0),\\
				\displaystyle \eta t-\frac{\sin{2\pi\eta t}}{2\pi}, & t \in[0,\frac{1}{\eta}],\\
				1, & t \in (\frac{1}{\eta},\infty).
			\end{cases}
			\label{eq:xit}
		\end{equation}
		\label{ex:Nt}
        \Cref{fig:nu vs t} shows \eqref{eq:xit} for $\eta \in \{1, \hspace{1mm} 2, \hspace{1mm} 5\}$. \hfill $\triangle$
\end{example}

 \vspace{0mm}
\begin{figure}[H]
    \centering
    \begin{subfigure}[t]{\linewidth} 
        \includegraphics[trim={0.5cm 0cm 1cm 0.25cm}, clip, width=\linewidth]{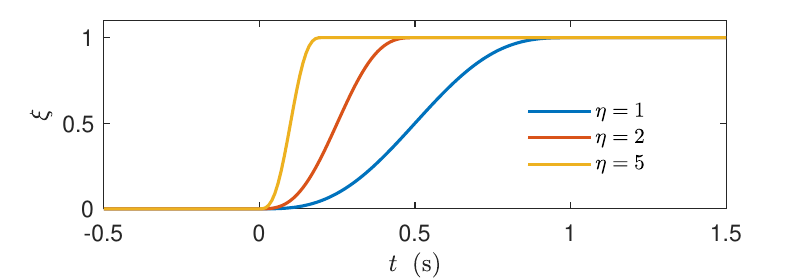}
    \end{subfigure}
    \caption{$\xi$ given by \Cref{ex:xi}.}
    \label{fig:nu vs t}
\end{figure}
 \vspace{0mm}

Define $\mathcal{T}_{k} \triangleq [t_{k}, t_{k+1})$.
For all $ k \in \mathbb{N}$ and all $t \in \mathcal{T}_{k}$, let $\theta:\mathcal{T}_{k} \rightarrow \mathbb{R}^{p}$ be given by
	\begin{equation}
		\theta(t) \triangleq
		\xi\bigg( \frac{t-t_{k}}{t_{k+1}-t_{k}}\bigg) \theta_{k} + \bigg[1-\xi\bigg( \frac{t-t_{k}}{t_{k+1}-t_{k}}\bigg)\bigg] \theta_{k-1},
		\label{eq:thetat}
	\end{equation}
 and let $\nu: \mathcal{T}_{k} \rightarrow [0,\infty)$ be given by
	\begin{equation}
		\nu (t) \triangleq
		\xi\bigg( \frac{t-t_{k}}{t_{k+1}-t_{k}}\bigg)\nu_{k}+\bigg[1-\xi\bigg( \frac{t-t_{k}}{t_{k+1}-t_{k}}\bigg)\bigg]\nu_{k-1},
		\label{eq:ct}
	\end{equation}
where $\nu_{-1} \triangleq \nu_{0}$ and $\theta_{-1} \triangleq \theta_{0}$.
Note that $\nu$ and $\theta$ are continuously differentiable functions constructed from the sequences $\nu_k$ and $\theta_k$.

Next, let $\alpha: \mathbb{R} \rightarrow \mathbb{R}$ be locally Lipschitz and nondecreasing such that $\alpha(0) = 0$, and consider the state-constraint function $\psi:\mathbb{R}^{n} \times \mathbb{R}^{p} \times \mathbb{R} \times \mathbb{R}^{m} \times \mathbb{R} \rightarrow \mathbb{R}$ be defined by
\begin{align}
	    \psi(x,\hat{\theta},\hat{\nu},\hat{u},\hat{\delta}) & \triangleq  L_{f}\psi_{d-1}(x) + L_{g}\psi_{d-1}(x)\hat{u} \nn \\
	    & \qquad +L_{\phi}\psi_{d-1}(x)\hat{\theta} - ||L_{\phi}\psi_{d-1}(x)||\hat{\nu} \nn \\
            & \qquad + \alpha_{d-1}(\psi_{d-1}(x)) + \hat{\delta} \psi_{d-1}(x), \label{eq:CBF-const}
\end{align}
where $\hat u$ is the control variable and $\hat \delta$ is a slack variable. 
Define 
\begin{equation}
\psi_*(x,\hat{u},\hat{\delta}) \triangleq \psi(x,\theta_{*},0,\hat{u},\hat{\delta}),
\label{eq:CBF_real}
\end{equation}
and note that $\psi_*(x(t),\hat{u},\hat{\delta}) \ge 0$ is a CBF-based state constraint for \eqref{eq:dyn} that guarantees for all $t \geq 0$, $x(t) \in \bar C$.
However, $\psi_*$ depends on the unknown parameter $\theta_{*}$. 
The next result shows that if $\hat \nu \ge \| \hat \theta -\theta_* \|$, then $\psi$ is a lower bound for $\psi_*$.
Since \ref{prop:tkck2:ck1} of \Cref{prop:tkck2} implies that $\nu(t) \ge \| \theta(t) -\theta_* \|$, the next result also demonstrates that the estimate $\theta$ and upper bound $\nu$ can be used in \eqref{eq:CBF-const} to obtain a constraint that is sufficient for $\psi_*(x(t),\hat{u},\hat{\delta}) \ge 0$.


\begin{proposition}\label{prop:prop3}\rm
Let $x\in \mathbb{R}^{n}$, $\hat{u} \in \mathbb{R}^{m}$, and $\hat{\delta} \in \mathbb{R}$. 
Then, the following hold:
\begin{enumerate}

\item \label{prop:prop3:a} 
Let $\hat{\theta} \in \mathbb{R}^{p}$ and let $\hat{\nu} \geq \| \hat{\theta}-\theta_{*} \|$. 
Then, $\psi(x,\hat{\theta},\hat{\nu},\hat{u},\hat{\delta}) \le \psi_*(x,\hat{u},\hat{\delta})$.

\item \label{prop:prop3:b} 
For all $t \geq 0$, $\psi(x,\theta(t),\nu (t),\hat{u},\hat{\delta}) \le \psi_*(x,\hat{u},\hat{\delta})$.

\item \label{prop:prop3:c} 
Assume there exists $k_{\rm{i}} \in \mathbb{N}$ such that for all $k \geq k_{i}$, $\Omega_k$ is full rank.
Then, $\lim_{t\to\infty} \psi(x,\theta(t),\nu (t),\hat{u},\hat{\delta}) = \psi_*(x,\hat{u},\hat{\delta})$
   
\end{enumerate}
\end{proposition}

 %
 %
 %

\begin{proof}
To prove \ref{prop:prop3:a}, since $\hat{\nu} \geq \| \hat{\theta}-\theta_{*} \|$, it follows that 
\begin{equation*}
L_{\phi}\psi_{d-1}(x) [\hat{\theta}-\theta_{*}] \le 
\| L_{\phi}\psi_{d-1}(x) \| \| \hat{\theta}-\theta_{*} \| \le \|L_{\phi}\psi_{d-1}(x)\| \hat{\nu} ,
\end{equation*}
which implies
\begin{equation*}
L_{\phi}\psi_{d-1}(x) \theta_{*} \geq L_{\phi}\psi_{d-1}(x) \hat{\theta}-\| L_{\phi}\psi_{d-1}(x)\| \hat{\nu}.
\end{equation*}
Adding $L_{f}\psi_{d-1}(x)+L_{g}\psi_{d-1}\hat{u}+\hat{\delta}\psi_{d-1}(x)+\alpha_{d-1}(\psi_{d-1}(x))$ to both sides yields $\psi_*(x,\hat{u},\hat{\delta}) \geq \psi(x,\hat{\theta},\hat{\nu},\hat{u},\hat{\delta})$.

To prove \ref{prop:prop3:b}, it follows from \eqref{eq:thetat} that 
\begin{equation}
\|\theta(t)-\theta_{*}\| \leq \xi \|\theta_{k}-\theta_{*}\|+\Big (1-\xi \Big ) \|\theta_{k-1}-\theta_{*}\|,
\label{eq:prop:prop3:b:eq1}
\end{equation}
where the argument of $\xi$ is omitted for brevity. 
Next, using \ref{prop:tkck2:ck1} of \Cref{prop:tkck2} and \eqref{eq:ct} yields
\begin{equation*}
\|\theta(t)-\theta_{*} \| \leq \xi \nu_{k}+\Big(1-\xi \Big)\nu_{k-1} = \nu (t).
\end{equation*}
Since $\nu(t) \ge \|\theta(t)-\theta_{*} \|$, it follows from 
\ref{prop:prop3:a} that $\psi_*(x,\hat{u},\hat{\delta}) \geq \psi(x,\theta(t),\nu (t),\hat{u},\hat{\delta})$.

To prove \ref{prop:prop3:c}, note that \eqref{eq:thetat} and \ref{prop:prop:tk2} of \Cref{prop:prop} imply that $\lim_{t \to \infty} \theta(t) = \theta_{*}$.
Similarly, \eqref{eq:ct} and \ref{prop:tkck2:ck4} of \Cref{prop:tkck2} imply that $\lim_{t \to \infty} \nu(t) = 0$.
Thus, it follows from \cref{eq:CBF-const,eq:CBF_real} that $\lim_{t \to \infty} \psi(x,\theta(t),\nu(t),\hat{u},\hat{\delta}) = \psi_{*}(x,\hat{u},\hat{\delta})$.
    \end{proof}

Next, let $\beta>0$, and consider the cost function $\bar{J}:\mathbb{R}^{n} \times \mathbb{R}^{p} \times \mathbb{R}^{m} \times \mathbb{R} \rightarrow \mathbb{R}$ given by
\begin{equation}
\bar{J}(x,{\theta},\hat{u},\hat{\delta})=J(x,\theta,\hat{u})+\frac{\beta}{2} \hat{\delta}^{2},
\label{eq:cost_func}
\end{equation}
which is equal to \eqref{eq:cost_func0} plus a term that weights the slack variable $\hat \nu$. 
For each $t \ge 0$, the objective is to synthesize $(\hat{u},\hat{\delta})$ that minimizes $\bar{J}(x(t),\theta(t),\hat{u},\hat{\delta})$ subject to the CBF state constraint $\psi(x(t),\theta(t),\nu(t),\hat{u},\hat{\delta}) \geq 0$.

For all $(x, \theta, \nu) \in \mathbb{R}^{n}\times \mathbb{R}^{p} \times \mathbb{R}$, the minimizer of $\bar{J}(x,{\theta},\hat{u},\hat{\delta})$ subject to $\psi(x,{\theta},{\nu},\hat{u},\hat{\delta})$ can be obtained from the first-order necessary conditions for optimality.
For example, see \cite{rabiee2024closed, safari2024time, ames2016control, wieland2007constructive, cortez2022compatibility}.
%
%
%
%
%
%
%
The first-order necessary conditions yield the control $u_{*}: \mathbb{R}^{n} \times \mathbb{R}^{p} \times \mathbb{R} \rightarrow \mathbb{R}^{m}$ defined by
\begin{equation}
u_{*}(x,{\theta},{\nu}) \triangleq u_{\rm{d}}(x,{\theta}) +\lambda_{*}(x,{\theta},{\nu}) H(x,{\theta})^{-1} L_{g}\psi_{d-1}(x)^\rmT,
\label{eq:control_alpha}
\end{equation}
and the slack variable $\delta_{*}: \mathbb{R}^{n} \times \mathbb{R}^{p} \times \mathbb{R} \rightarrow \mathbb{R}$ given by
	\begin{equation}
		\delta_{*}(x,{\theta},{\nu}) \triangleq \beta^{-1} \psi_{d-1}(x) \lambda_{*}(x,{\theta},{\nu}), 
		\label{eq:control_delta}
	\end{equation}
where $\lambda_{*}: \mathbb{R}^{n} \times \mathbb{R}^{p} \times \mathbb{R} \rightarrow \mathbb{R}$ is defined by
	\begin{equation}
		\lambda_{*}(x,{\theta},{\nu}) \triangleq
		\begin{cases} 
			\displaystyle -\frac{\omega(x,{\theta},{\nu})}{q(x,{\theta})}, & \omega(x,{\theta},{\nu})<0,\\
			0, & \omega(x,{\theta},{\nu}) \geq 0,
		\end{cases}
		\label{eq:control_lambda}
	\end{equation}
	and $q: \mathbb{R}^{n} \times \mathbb{R}^{p} \rightarrow \mathbb{R}$ and $\omega: \mathbb{R}^{n} \times \mathbb{R}^{p} \times \mathbb{R} \rightarrow \mathbb{R}$ are given by
        \begin{equation}
            \begin{aligned}
                q(x,{\theta}) &\triangleq L_{g}\psi_{d-1}(x)H(x,{\theta})^{-1}L_{g}\psi_{d-1}(x)^{\rmT}\\
                & \qquad +\beta^{-1} \psi_{d-1}(x)^{2},
            \end{aligned}
        \label{eq:dx}
        \end{equation}
        \begin{equation}
    	\omega(x,{\theta},{\nu}) \triangleq \psi(x,{\theta},{\nu},u_{\rm{d}}(x,{\theta}),0).
    	\label{eq:omegax}
        \end{equation}

The next result shows that $(u_{*}(x,{\theta},{\nu}),\delta_{*}(x,{\theta},{\nu}))$ is the unique global minimizer of $\bar{J}(x,{\theta},\hat{u},\hat{\delta})$ subject to $\psi(x,{\theta},{\nu},\hat{u},\hat{\delta})$.
The proof is similar to that of \cite[Theorem 1]{safari2024time} and is omitted for space.

\begin{theorem}\label{th:cf_optimality}\rm 
Assume Assumption~\ref{assum:assump2} is satisfied. 
Let $(x,\theta,\nu) \in \mathbb{R}^{n}\times \mathbb{R}^{p}\times \mathbb{R}$, and let $(\hat u, \hat \delta) \in \mathbb{R}^{m}\times \mathbb{R}$ be such that $(\hat u, \hat \delta) \neq (u_{*}(x,{\theta},{\nu}),\delta_{*}(x,{\theta},{\nu})  )$ and $\psi(x,{\theta},{\nu},\hat{u},\hat{\delta}) \geq 0$. 
Then, $$\bar{J}(x,{\theta},\hat{u},\hat{\delta})>\bar{J}(x,{\theta},u_{*}(x,{\theta},{\nu}),\delta_{*}(x,{\theta},{\nu})).$$
\end{theorem}

The following theorem is the main result on satisfaction of the state constraint despite the model uncertainty. 
This result follows from standard CBF analysis techniques (e.g., \cite{ames2016control}) and is omitted for space.

\begin{theorem}\label{th:main}\rm 
Consider \eqref{eq:dyn}, where Assumptions~\ref{assum:input} and~\ref{assum:assump2} are satisfied. 
Let $u=u_{*}$, where $u_{*}$ is given by \cref{eq:control_alpha,eq:control_delta,eq:control_lambda,eq:dx,eq:omegax}, where $\theta$ and $\nu$ are given by \cref{eq:thetat,eq:ct,eq:tauk,eq:yhat,eq:ck1,eq:theta_k,eq:phik,eq:Omega,eq:Pk,eq:yk}.
Assume that $\psi_{d-1}^\prime$ is locally Lipschitz. 
Then, for all $x_0 \in \bar \SC$, the following statements hold: 
\begin{enumerate}
    \item There exists a maximum value $t_{\rm m} (x_0 ) \in (0 ,\infty]$ such that \eqref{eq:dyn} with $u = u_*$ has a unique solution on $[0, t_\rm (x_0))$.
    
    \item For all $t \in [0, t_\rm (x_0))$, $x(t) \in \bar \SC$.

    \item Assume the maximum interval of existence and uniqueness is $t_\rm (x_0) =\infty$. 
    Then, for all $t \ge 0$, $x(t) \in \bar \SC$.

\end{enumerate}
\end{theorem}

\Cref{th:cf_optimality,th:main} demonstrate that the control \cref{eq:control_alpha,eq:control_delta,eq:control_lambda,eq:dx,eq:omegax} satisfies the state constraint that for all $t \ge 0$, $x(t) \in \SC_0$, and yields a control that is optimal with respect to a constraint optimization that aims to obtain a control as close as possible to $u_\rmd$ subject to the state constraint.

    \section{Inverted pendulum}
    Consider the pendulum modeled by \eqref{eq:dyn}, where
    \begin{equation}
    f(x) =
    \left[
    \begin{array}{c}
        \dot{\gamma} \\
        f_{2}(x)
    \end{array}
    \right],
    \hspace{1mm}
    g(x)=
    \left[
    \begin{array}{c}
        0\\
        \frac{1}{mL^{2}} 
    \end{array}
    \right],
    \hspace{1mm}
    \phi(x)=
    \left[
    \begin{array}{c}
        0 \\
        \phi_{2}(x)
    \end{array}
    \right],
    \nonumber
    \end{equation}
    where 
    \begin{equation}
        \displaystyle
        \phi_{2}(x) \triangleq \frac{1}{mL^{2}}\bigg[-\gamma \hspace{2mm} -\gamma^{3} \hspace{2mm} -\tanh{\frac{\dot{\gamma}}{\epsilon_{1}}} \hspace{2mm} -\dot{\gamma} \hspace{2mm} -\dot{\gamma}^{2}\tanh{\frac{\dot{\gamma}}{\epsilon_{2}}} \bigg],
        \nonumber
    \end{equation}
    \begin{equation}
        f_{2}(x) \triangleq \frac{g}{L}\sin{\gamma}, \hspace{5mm}
         x = \left[
        \begin{array}{c}
            \gamma\\
            \dot{\gamma}
        \end{array}
        \right],
        \nonumber
    \end{equation}
    and $\gamma$ is the position, $\dot{\gamma}$ is the velocity, $g$ is the gravity, $\epsilon_{1}=\epsilon_{2} =2$, $m=0.01$ kg and $L=0.15$ m. In $\phi_{2}(x)$ the first two entries represent linear and cubic restitution forces, while the last three entries represent coulomb friction, viscous friction and drag force. Furthermore, $\theta_{*}=[0.5  \hspace{2mm} 0.35 \hspace{2mm} 0.15 \hspace{2mm} 0.5 \hspace{2mm} 0.25]^{\rm{T}} \in \Theta$, where
    \begin{equation}
        \Theta = \left\{\hat{\theta} \in \mathbb{R}^{5} \mid \text{for all } j \in \{1, 2, \ldots, 5\}, \, e_{j}^{\rm{T}}\hat{\theta} \in [0 \hspace{1mm} , \hspace{1mm} 2.5] \right\},
    \end{equation}
    and  $e_j$ is the $jth$ column of the $p \times p$ identity matrix.

    We implement \cref{eq:thetat,eq:ct,eq:tauk,eq:yhat,eq:ck1,eq:theta_k,eq:phik,eq:Omega,eq:Pk,eq:yk}, where $t_{k+1}-t_{k}= 0.25$ s, $k_{\rm{n}}=30$, $\theta_{0}=0$, $\nu_{0}=\sup_{\hat{\theta}\in \Theta}||\theta_{0}-\hat{\theta}||$, $\sigma_{k}=0.1$ and  $\eta=2$.
    
    Let $\displaystyle \gamma_{\rm{d}}(t)\triangleq-0.99(\pi/4)\cos{t}$ be the desired angular position. Define the error $e \triangleq \gamma-\gamma_{\rm{d}}$ and the desired control
    \begin{equation}
        u_{\rm{d}}(\hat{\theta}) \triangleq mL^{2}\bigg[-f_{2}-\phi_{2}\hat{\theta}+\dot{\omega}_{d}-K_{1}e-K_{2}\dot{e}\bigg],
        \nonumber
    \end{equation}
    where $K_{1}=50$ and $K_{2}=100$.

    If $u(t)=u_{\rm{d}}(\theta_{*})$, then $\ddot{e}+K_{1}\dot{e}+K_{2}e=0$, which implies that $\lim_{t \rightarrow \infty} e(t)=0$ and $\lim_{t\rightarrow \infty} \dot{e}(t)=0$.
    
    The safe set is given by \eqref{eq:psi0}, where $d=2$ and
    \begin{equation}
        \psi_{0}(x) = (\frac{\pi}{4})^{2}-\gamma^{2}.
        \label{eq:pend_psi0}
    \end{equation}
    
    We implement the control \cref{eq:control_alpha,eq:control_delta,eq:control_lambda,eq:dx,eq:omegax}, where $H=2$, $\beta=200$ and $\alpha_{0}=\alpha_{1}=200$. The control is updated at 1000 Hz using a zero-order hold.
    
    To examine the impact of the adaptive estimate $\theta$ and the adaptive bound $\nu$, we consider 3 cases: 
    
    \begin{enumerate}[label=\alph*)]
        \item Adaptive estimate $\theta$ and adaptive bound $\nu$ are used in CBF constraint and desired control (i.e., $\psi=\psi(x,\theta,\nu,\hat{u},\hat{\delta})$, $u_{\rm{d}}=u_{\rm{d}}(\theta)$).
        \item Adaptive estimate $\theta$ and adaptive bound $\nu$ are used in CBF constraint (i.e., $\psi=\psi(x,\theta,\nu,\hat{u},\hat{\delta})$), but the desired control uses initial estimate (i.e., $u_{\rm{d}}=u_{\rm{d}}(\theta_{0})$).
        \item Initial estimate $\theta_{0}$ and initial bound $\nu_{0}$ are used in CBF constraint (i.e., $\psi=\psi(x,\theta_{0},\nu_{0},\hat{u},\hat{\delta})$), but the desired control uses the adaptive estimate (i.e., $u_{\rm{d}}=u_{\rm{d}}(\theta)$).
    \end{enumerate}
    \vspace{-0mm}
    Proposition \ref{prop:prop3} implies that all 3 cases satisfy the state constraint. However, Cases 2 and 3 use the initial estimates, which can be conservative and may lead to worse performance, that is, worse tracking of $\gamma_{\rm{d}}$.

    \vspace{-2mm}
    \begin{figure}[H]
    \centering
    \begin{subfigure}[t]{0.32\linewidth} 
        \includegraphics[trim={2.22cm 2.1cm 18.8cm 4.1cm},clip, width=1.2\linewidth]{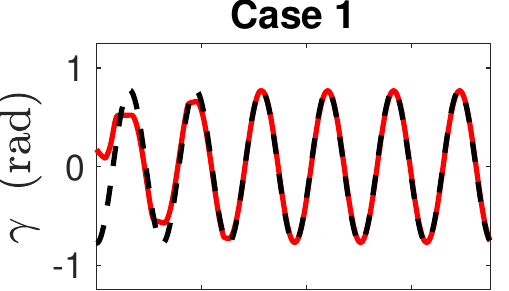}
    \end{subfigure}
    \hspace{0.22cm}
    \begin{subfigure}[t]{0.32\linewidth} 
        \includegraphics[trim={10.56cm 2.1cm 12.2cm 4.1cm},clip, width=0.96\linewidth]{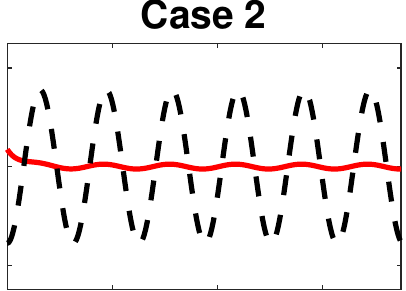}
    \end{subfigure}
    \hspace{-0.35cm}
    \begin{subfigure}[t]{0.32\linewidth} 
        \includegraphics[trim={17.41cm 2.1cm 5.4cm 4.1cm},clip, width=0.95\linewidth]{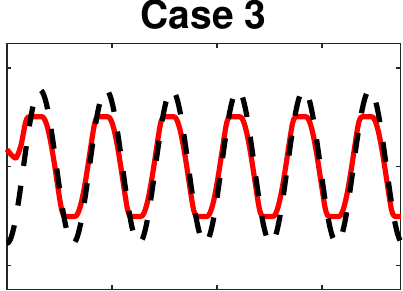}
    \end{subfigure}\\
    \begin{subfigure}[t]{0.32\linewidth} 
        \includegraphics[trim={2.22cm 2.1cm 18.8cm 4.8cm},clip, width=1.2\linewidth]{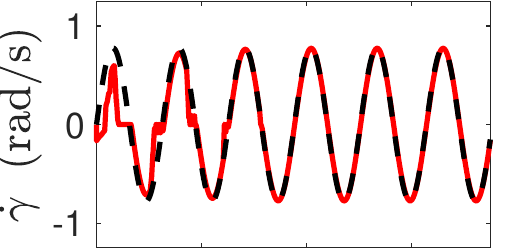}
    \end{subfigure}
    \hspace{0.22cm}
    \begin{subfigure}[t]{0.32\linewidth} 
        \includegraphics[trim={10.56cm 2.1cm 12.2cm 4.8cm},clip, width=0.96\linewidth]{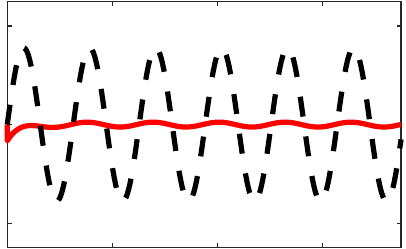}
    \end{subfigure}
    \hspace{-0.35cm}
    \begin{subfigure}[t]{0.32\linewidth} 
        \includegraphics[trim={17.41cm 2.1cm 5.4cm 4.8cm},clip, width=0.95\linewidth]{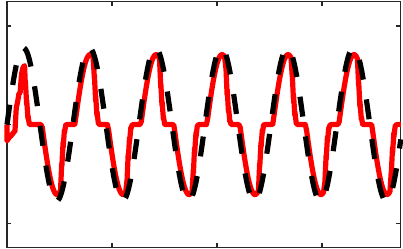}
    \end{subfigure}\\
    \begin{subfigure}[t]{0.32\linewidth} 
        \includegraphics[trim={2.22cm 0.4cm 18.8cm 4.8cm},clip, width=1.2\linewidth]{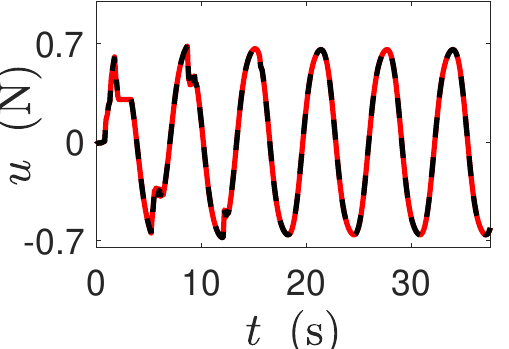}
    \end{subfigure}
    \hspace{0.22cm}
    \begin{subfigure}[t]{0.32\linewidth} 
        \includegraphics[trim={10.56cm 0.4cm 12.2cm 4.8cm},clip, width=0.96\linewidth]{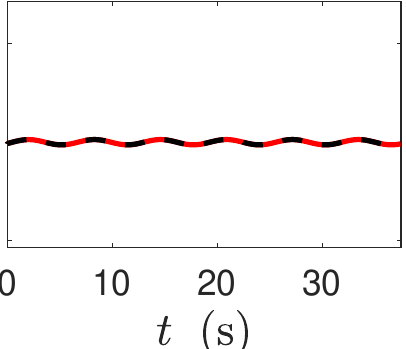}
    \end{subfigure}
    \hspace{-0.35cm}
    \begin{subfigure}[t]{0.32\linewidth} 
        \includegraphics[trim={17.41cm 0.4cm 5.4cm 4.8cm},clip, width=0.95\linewidth]{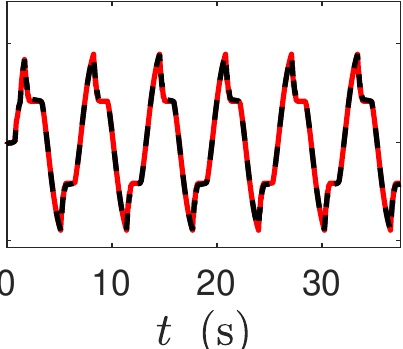}
    \end{subfigure}
    \caption{$\gamma$, $\dot{\gamma}$ and $u$ for Cases 1, 2 and 3. Note that $\gamma_{\rm{d}}$, $\dot{\gamma}_{\rm{d}}$ and $u_{\rm{d}}$ are shown with dashed lines.}
    \label{fig:pend:state+control}
    \end{figure}
    \vspace{-2mm}

    Figure \ref{fig:pend:state+control} shows $\gamma$, $\dot{\gamma}$ and $u$, where $x(0)=[0.1745 \hspace{2mm} 0]^{\rm{T}}$. In Case 1 it is shown that $\lim_{t \rightarrow \infty} \gamma(t) = \gamma_{d}$ and $\lim_{t \rightarrow \infty} \dot{\gamma}(t) = \dot{\gamma}_{d}$, unlike in Case 2. In Case 3, $\gamma$ and $\dot{\gamma}$ diverges from $\gamma_{d}$ and $\omega_{d}$ near to $\gamma = \pm \frac{\pi}{4}$ rad.

    Figure \ref{fig:pend:safety} shows that, for all  $t\geq0$, $\psi_{0}$, $\psi_{1}$ and $\psi$ remain nonnegative in all cases. In Case 1, the safety constraint is not active in steady-state.  In Case 2, for all $t \geq 0$, $\psi(x(t)) > 0$. Finally, in Case 3, the safety constraint is periodically activated due to its inherent conservativeness, explaining why $\gamma$ and $\dot{\gamma}$ diverge from $\gamma_{d}$ and $\omega_{d}$ near to $\gamma = \pm \frac{\pi}{4}$ rad.
    

    \begin{figure}[H]
    \centering
    \begin{subfigure}[t]{0.32\linewidth} 
        \includegraphics[trim={1.6cm 2.1cm 18.8cm 4.1cm},clip, width=1.245\linewidth]{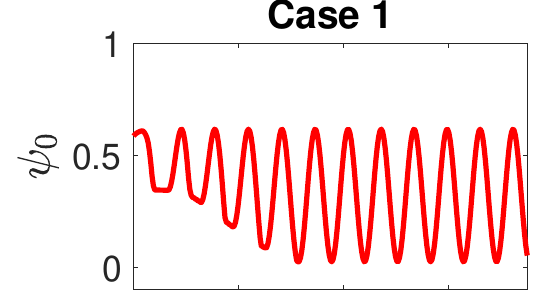}
    \end{subfigure}
    \hspace{0.34cm}
    \begin{subfigure}[t]{0.32\linewidth} 
        \includegraphics[trim={10.56cm 2.1cm 12.2cm 4.1cm},clip, width=0.925\linewidth]{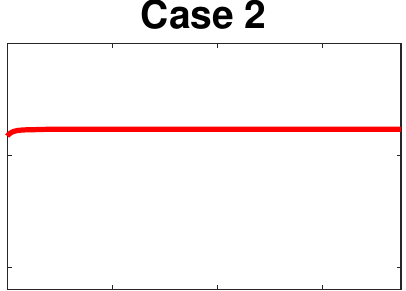}
    \end{subfigure}
    \hspace{-0.44cm}
    \begin{subfigure}[t]{0.32\linewidth} 
        \includegraphics[trim={17.41cm 2.1cm 5.4cm 4.1cm},clip, width=0.925\linewidth]{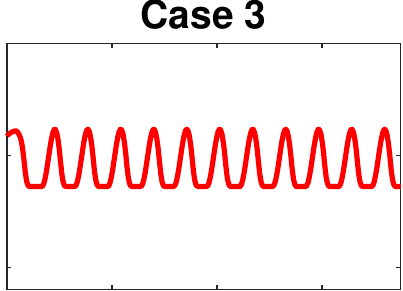}
    \end{subfigure}\\
    \begin{subfigure}[t]{0.32\linewidth} 
        \includegraphics[trim={1.6cm 2.1cm 18.8cm 4.8cm},clip, width=1.245\linewidth]{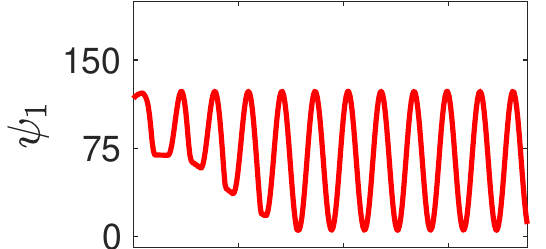}
    \end{subfigure}
    \hspace{0.34cm}
    \begin{subfigure}[t]{0.32\linewidth} 
        \includegraphics[trim={10.56cm 2.1cm 12.2cm 4.8cm},clip, width=0.925\linewidth]{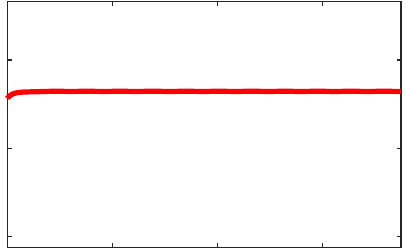}
    \end{subfigure}
    \hspace{-0.44cm}
    \begin{subfigure}[t]{0.32\linewidth} 
        \includegraphics[trim={17.41cm 2.1cm 5.4cm 4.8cm},clip, width=0.925\linewidth]{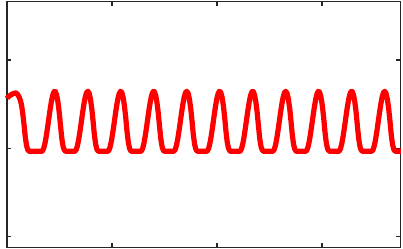}
    \end{subfigure}\\
    \begin{subfigure}[t]{0.32\linewidth} 
        \includegraphics[trim={1.6cm 0.4cm 18.8cm 4.8cm},clip, width=1.245\linewidth]{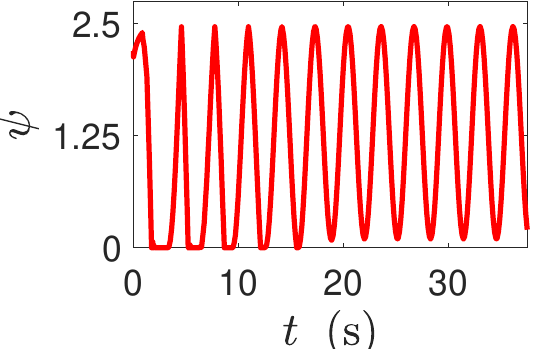}
    \end{subfigure}
    \hspace{0.34cm}
    \begin{subfigure}[t]{0.32\linewidth} 
        \includegraphics[trim={10.56cm 0.4cm 12.2cm 4.8cm},clip, width=0.925\linewidth]{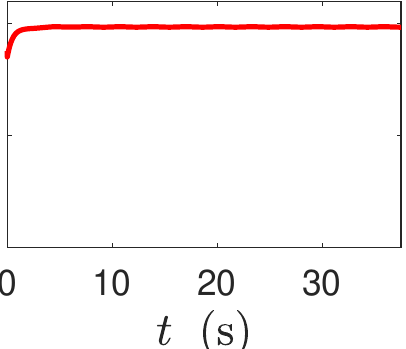}
    \end{subfigure}
    \hspace{-0.44cm}
    \begin{subfigure}[t]{0.32\linewidth} 
        \includegraphics[trim={17.41cm 0.4cm 5.4cm 4.71cm},clip, width=0.925\linewidth]{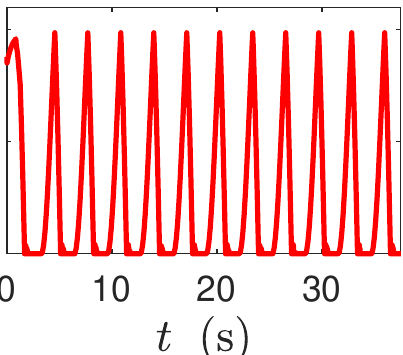}
    \end{subfigure}
    \caption{$\psi_{0}$, $\psi_{1}$ and $\psi$ for Cases 1, 2 and 3. }
    \label{fig:pend:safety}
    \end{figure}

    Figures \ref{fig:pend:est} and \ref{fig:pend:est_error} show that for Cases 1 and 3, $\lim_{t\rightarrow \infty} \theta(t) = \theta_{*}$ and $\lim_{t\rightarrow \infty} \nu (t) = 0$, unlike in Case 2.

\section{Nonholonomic mobile robot}

    
    Consider the nonholonomic differential drive mobile robot modeled by \eqref{eq:dyn}, where
    \begin{equation}
        f(x)=
        \left[\hspace{-1mm}
        \begin{array}{c}
            v\cos{\gamma}-l_{d}\omega\sin{\gamma} \\
            v\sin{\gamma}+l_{d}\omega\cos{\gamma} \\
            \omega \\
            0 \\
            0
        \end{array}
        \hspace{-1mm}\right],  \hspace{1mm}
        g(x)=
        \left[\hspace{-1mm}
        \begin{array}{cc}
            0 & 0 \\
            0 & 0 \\
            0 & 0\\
            g_{1}(x) & g_{2}(x)
        \end{array}
        \hspace{-1mm}\right],
        \nonumber
    \end{equation} 
    \begin{equation}
    \phi(x)=
    \left[\hspace{-1mm}
    \begin{array}{c}
        0 \\
        0 \\
        0 \\
        \phi_{4}(x) \\
        \phi_{5}(x)
    \end{array}
    \hspace{-1mm}\right], \hspace{0mm}
    x = \left[
        \begin{array}{c}
            q_{\rm{x}}\\
            q_{\rm{y}}\\
            \gamma\\
             v \\
             \omega
        \end{array}
        \right], \hspace{0mm}
        u = \left[\hspace{-1mm}
        \begin{array}{c}
            u_{\rm{r}}\\
            u_{\rm{l}}
        \end{array}
        \hspace{-1mm}\right], \hspace{0mm}
        \theta_{*}=\left[\hspace{-1mm}
        \begin{array}{c}
            k_{\rm{b1}}\\
            k_{\rm{b2}}\\
            \epsilon \\
            \kappa
        \end{array}
        \hspace{-1mm}\right],
    \nonumber
    \end{equation}
    where 
    \begin{equation}
    \phi_{4}(x) \triangleq \bigg[\displaystyle -\frac{k_{\rm{m}}\omega_{\rm{r}}}{mrR_{\rm{a}}} \hspace{2mm} \displaystyle -\frac{k_{\rm{m}}\omega_{\rm{l}}}{mrR_{\rm{a}}}  \hspace{2mm} \displaystyle -\frac{\omega_{\rm{r}}}{mr}-\frac{\omega_{\rm{l}}}{mr}  \hspace{2mm} \displaystyle  -g\sin{\gamma}\bigg],
    \nonumber
    \end{equation}
    \begin{equation}
    \phi_{5}(x) \triangleq
    \bigg[\displaystyle -\frac{k_{\rm{m}}l\omega_{\rm{r}}}{IrR_{\rm{a}}} \hspace{2mm} \displaystyle \frac{k_{\rm{m}}l\omega_{\rm{l}}}{IrR_{\rm{a}}} \hspace{2mm} \displaystyle -\frac{\omega_{\rm{r}}l}{Ir}+\frac{\omega_{\rm{l}}l}{Ir} \hspace{2mm} \displaystyle  0\bigg],
    \nonumber
    \end{equation} 
    \begin{equation}
        g_{1}(x)\triangleq
        \left[\hspace{-2mm}
        \begin{array}{cc}
            \displaystyle \frac{k_{\rm{m}}}{mrR_{\rm{a}}} & \displaystyle \frac{k_{\rm{m}}l}{IrR_{\rm{a}}}
        \end{array}
        \hspace{-2mm}\right]^{T}\hspace{-2mm}, \hspace{2mm}
        g_{2}(x)\triangleq
        \left[\hspace{-2mm}
        \begin{array}{cc}
            \displaystyle \frac{k_{\rm{m}}}{mrR_{\rm{a}}} & \displaystyle -\frac{k_{\rm{m}}l}{IrR_{\rm{a}}}
        \end{array}
        \hspace{-2mm}\right]^{T},
        \nonumber
    \end{equation} 
    \begin{equation}
         \omega_{\rm{r}}\triangleq \frac{2v+l\omega}{2r} \hspace{1cm},  \hspace{2mm} \hspace{1cm} \omega_{\rm{l}}\triangleq\frac{2v-l\omega}{2r},
        \nonumber
    \end{equation}
     and $[q_{\rm{x}} \hspace{2mm} q_{\rm{y}}]^{\rm{T}}$ denote the position of the tip of the robot in an orthogonal coordinate frame, $\gamma$ is the direction of the velocity vector, $v$ and $\omega$ are the velocity and angular velocity, $u_{\rm{r}}$ and $u_{\rm{l}}$ are the voltage of each motor, $k_{\rm{m}}=0.1$ N.m/Amp is the torque constant, $r=0.1$ m is the wheel radius, $l=0.5$ m is the distance between wheels, $l_{d}=0.25$ m is the distance from the center of mass to the tip of the vehicle, $R_{\rm{a}}=0.27$ ohms is the armature resistance, $m=10$ kg and $I=0.83$ $\rm{kg.m^{2}}$ are the vehicle mass and inertia, $\omega_{\rm{r}}$ and $\omega_{\rm{l}}$ are the angular velocity of the each wheel and $g$ is the gravity \cite{anvari2013non}. Moreover, $k_{\rm{b1}} = k_{\rm{b2}} = 0.0487$ V/(rad/sec) are the back-EMF constants of each motor, $\epsilon = 0.025 $ is a friction coefficient, and $\kappa = 0.5$ corresponds to an angle of inclination of the ground of $30 ^{\circ}$.
    
    \vspace{-2mm}
    \begin{figure}[H]
    \centering
    \begin{subfigure}[t]{0.32\linewidth} 
        \includegraphics[trim={1.85cm 0.2cm 19.25cm 0cm},clip, width=1.2\linewidth]{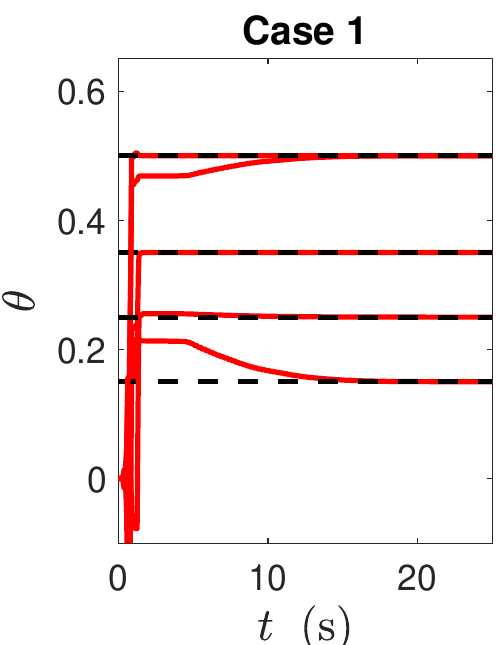}
    \end{subfigure}
    \hspace{0.3cm}
    \begin{subfigure}[t]{0.32\linewidth} 
        \includegraphics[trim={12cm 0.2cm 11.1cm 0cm},clip, width=0.92\linewidth]{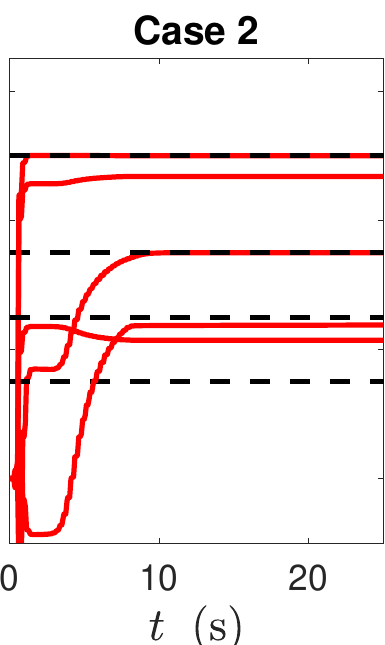}
    \end{subfigure}
    \hspace{-0.4cm}
    \begin{subfigure}[t]{0.32\linewidth} 
        \includegraphics[trim={20.3cm 0.2cm 2.8cm 0cm},clip, width=0.92\linewidth]{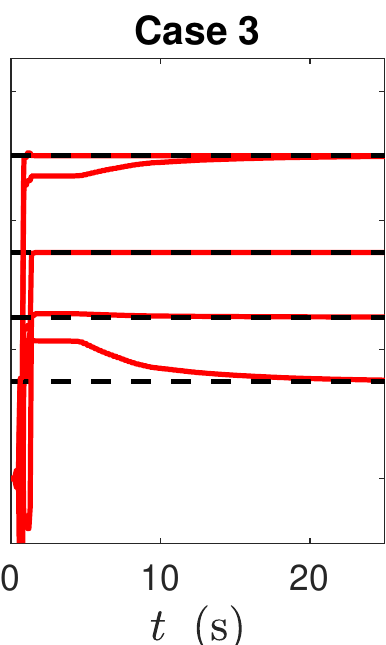}
    \end{subfigure}
    \caption{$\theta$ for Cases 1, 2 and 3. Note that $\theta_{*}$ is shown using dashed line.}
    \label{fig:pend:est}
    \end{figure}
    \vspace{-2mm}

    \vspace{-2mm}
    \begin{figure}[H]
    \centering
    \begin{subfigure}[t]{0.32\linewidth} 
        \includegraphics[trim={2.6cm 0.2cm 18.8cm 4.1cm},clip, width=1.17\linewidth]{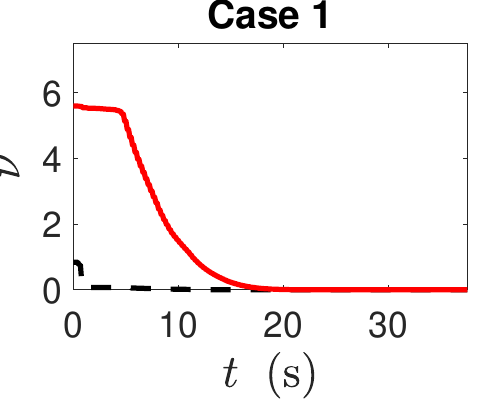}
    \end{subfigure}
    \hspace{0.15cm}
    \begin{subfigure}[t]{0.32\linewidth} 
        \includegraphics[trim={10.57cm 0.2cm 12.2cm 4.1cm},clip, width=0.975\linewidth]{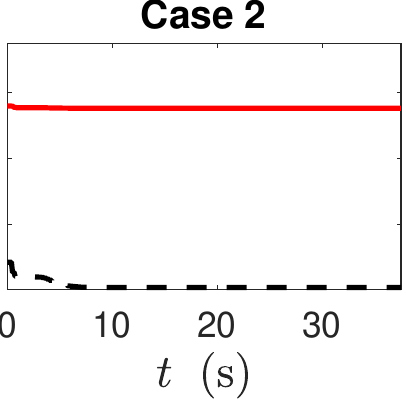}
    \end{subfigure}
    \hspace{-0.33cm}
    \begin{subfigure}[t]{0.32\linewidth} 
        \includegraphics[trim={17.4cm 0.2cm 5.4cm 4.1cm},clip, width=0.972\linewidth]{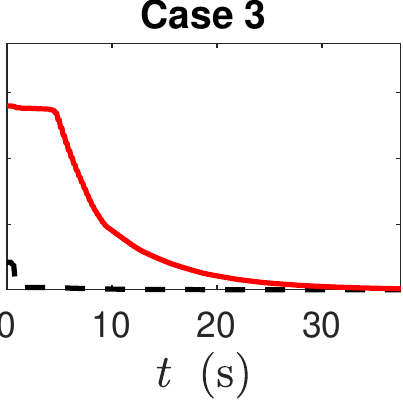}
    \end{subfigure}
    \caption{$\nu$ for Cases 1, 2 and 3. Note that $||\tilde{\theta}_{k}||$ is shown with dashed line. }
    \label{fig:pend:est_error}
    \end{figure}
    \vspace{-2mm}

    \vspace{-2mm}
    Let $\theta_{*}$ be in the set
    \begin{equation}
        \begin{aligned}
            \Theta = \left\{ \hat{\theta} \in \mathbb{R}^{4} \mid e_{1}^{\rm{T}}\hat{\theta} \in [0, 5], e_{2}^{\rm{T}}\hat{\theta} \in [0, 5], \right. \\
            \left. e_{3}^{\rm{T}}\hat{\theta} \in [0, 5], e_{4}^{\rm{T}}\hat{\theta} \in [0, 1] \right\}.
        \end{aligned}
        \nonumber
    \end{equation}    

    We implement \cref{eq:thetat,eq:ct,eq:tauk,eq:yhat,eq:ck1,eq:theta_k,eq:phik,eq:Omega,eq:Pk,eq:yk}, where $t_{k+1}-t_{k}= 0.1$ s, $k_{\rm{n}}=10$, $\theta_{0}=[0.1 \hspace{2mm} 0.1 \hspace{2mm} 0.1 \hspace{2mm} 0.1]^{\rm{T}}$, $\nu_{0}=\sup_{\hat{\theta}\in \Theta}||\theta_{0}-\hat{\theta}||$, $\sigma_{k}=0.001$ and $\eta=2$.
    
    The desired linear and angular velocity  \cite{rabiee2023automatica} are
    \begin{equation}
        v_{\rm{d}}\triangleq -\big(\mu_{1}+\mu_{2}\big)v-\big(1+\mu_{1}\mu_{2}\big)e_{1}+\frac{\mu_{1}^{2}}{l_{d}}e_{2}^{2}, \hspace{5mm} \omega_{\rm{d}}\triangleq -\frac{\mu_{1}}{l_{d}}e_{2},
        \nonumber
    \end{equation}
    where
    \begin{equation}
        e_{1}\triangleq (q_{\rm{x}}-q_{\rm{dx}})\cos{\gamma}+(q_{\rm{y}}-q_{\rm{dy}})\sin{\gamma},
        \nonumber
    \end{equation}
    \begin{equation}
        e_{2}\triangleq -(q_{\rm{x}}-q_{\rm{dx}})\sin{\gamma}+(q_{\rm{y}}-q_{\rm{dy}})\cos{\gamma},
        \nonumber
    \end{equation}
     $\mu_{1}=\mu_{2} =0.08$ and $q_{\rm{d}}=[q_{\rm{dx}} \hspace{2mm} q_{\rm{dy}}]^{\rm{T}}$ is the desired position of the tip the robot, with $q_{\rm{dx}} = 2.56$ m and $q_{\rm{dy}}=1.8$ m.
    
    Define $e_{\rm{a}} \triangleq v-v_{\rm{d}}$ and $e_{\rm{b}} \triangleq  \omega-\omega_{\rm{d}}$. Also, define
    \begin{multline}
        u_{\rm{d1}}(\hat{\theta})  \triangleq -\phi_{4}\hat{\theta} \hspace{2mm} +\\
        \frac{1}{1+\mu_{1}+\mu_{2}}\bigg(\frac{2\mu_{1}^{2}}{l_{d}}e_{2}\dot{e}_{2}-(1+\mu_{1}+\mu_{2})\dot{e}_{1}-K_{1}e_{\rm{a}} \bigg),
        \nonumber
    \end{multline}
    and
    \begin{equation}
        u_{\rm{d2}}(\hat{\theta})  \triangleq -\phi_{5}\hat{\theta} - \frac{\mu_{1}}{l_{d}}\dot{e}_{2}-K_{2}e_{\rm{b}},
        \nonumber
    \end{equation}
    where $K_{1}=10$ and $K_{2}=10$.
    
    Hence, the desired control is defined by $u_{\rm{d}}(\hat{\theta}) \triangleq [u_{\rm{dr}}(\hat{\theta}) \hspace{2mm} u_{\rm{dl}}(\hat{\theta})]^{\rm{T}}$, where 
    \begin{equation}
        u_{\rm{dr}}(\hat{\theta}) \triangleq
        \displaystyle \frac{mrR_{\rm{a}}}{2k_{\rm{m}}}u_{\rm{d1}}(\hat{\theta})+\frac{IrR_{\rm{a}}}{2k_{\rm{m}}l}u_{\rm{d2}}(\hat{\theta}),
        \nonumber
    \end{equation}
    \begin{equation}
        u_{\rm{dl}}(\hat{\theta}) \triangleq
        \displaystyle \frac{mrR_{\rm{a}}}{2k_{\rm{m}}}u_{\rm{d1}}(\hat{\theta})-\frac{IrR_{\rm{a}}}{2k_{\rm{m}}l}u_{\rm{d2}}(\hat{\theta}).
        \nonumber
    \end{equation}

    If $u(t)=u_{\rm{d}}(\theta_{*})$, then $\dot{v}=\dot{v}_{\rm{d}}-K_{1}(v-v_{\rm{d}})$ and $\dot{\omega}=\dot{\omega}_{\rm{d}}-K_{2}(\omega-\omega_{\rm{d}})$, which implies that $\lim_{t \rightarrow \infty} v(t) = v_{d}$ and $\lim_{t \rightarrow \infty} \omega(t) = \omega_{d}$.

    For $i\in\{1,2\}$, define
    \begin{equation}
        \displaystyle
        h_{i}(x) \triangleq 0.5\big(q_{\rm{x}}-c_{i,1}\big)^{2}+0.5\big(q_{\rm{y}}-c_{i,2}\big)^2-0.5\big(R_{i}+l_{d})^2
        \nonumber
    \end{equation}
    where $c_{1,1}=0.65$ m, $c_{1,2}=0.8$ m, $R_{1}=0.5$, $c_{2,1}=1.95$ m, $c_{2,2}=1.75$ m and $R_{2}=0.35$ m describe the physical obstacles, while $l_{d}$ accounts for the robot's dimension.

    The safe set is given by \eqref{eq:psi0}, where 
    \begin{equation}
        \psi_{0}(x) =-\frac{1}{\rho}\mathrm{log}\bigg(\sum_{i=1}^{2}\mathrm{\exp}\{-\rho h_{i}(x)\}\bigg),
        \nonumber
    \end{equation}
    and $\rho=3$ and $d=2$ \cite{rabiee2024closed}.

    We implement \cref{eq:control_alpha,eq:control_delta,eq:control_lambda,eq:dx,eq:omegax}, where $H=2I_{2}$, $\beta=20$ and $\alpha_{0}=5$ and $\alpha_{1}=2$. The control is updated at $200$ Hz using a zero-order hold structure.

    The cases of the first examples are revisited to analyze and show the impact of the estimate $\theta$ and adaptive bound $\nu$ in Case 1, and to examine the conservative safety constraint and the degraded performance in Cases 2 and 3.

    \begin{figure}[H]
        \centering
        \includegraphics[trim={4cm 8.4cm 4cm 9cm}, clip, width=\linewidth]{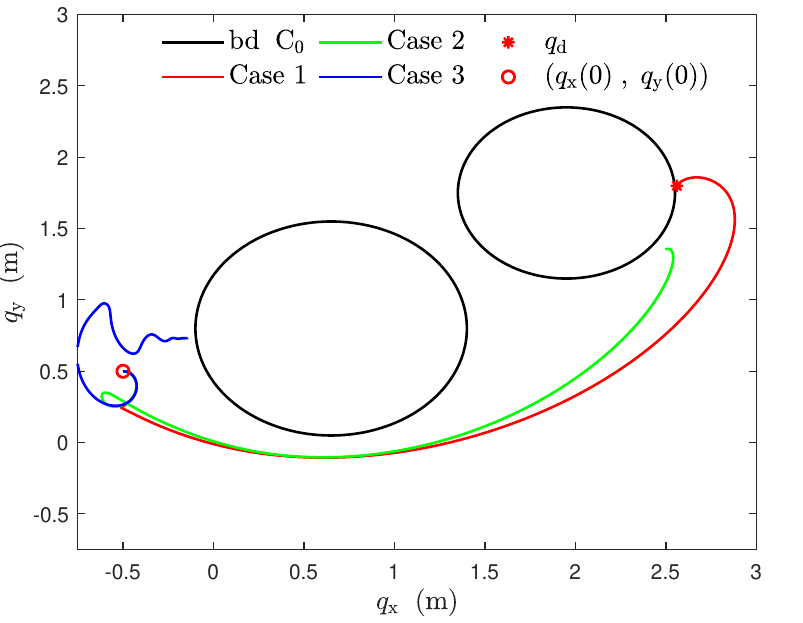}
        \caption{Phase-Portrait of 3 closed-loop trajectories for Cases 1, 2 and 3.}
        \label{fig:GR:phase_port}
    \end{figure}
    
    Figure \ref{fig:GR:phase_port} shows the closed-loop trajectories for $x_{0}=[-0.5 \hspace{2mm} 0.5 \hspace{2mm} 0 \hspace{2mm} 0 \hspace{2mm} 0]^{\rm{T}}$ with a goal at $q_{\rm{d}}=[2.56 \hspace{2mm} 1.8]^{\rm{T}}$ and the cases described above. 

   Figure \ref{fig:GR:states} shows the states for all 3 cases. Note that only for Case 1, $q_{\rm{x}}$, $q_{\rm{y}}$, $v$ and $\omega$ are driven to $q_{\rm{dx}}$, $q_{\rm{dy}}$, $v_{d}$ and $\omega_{d}$. However, in Cases 2 and 3 there is an offset in steady-state between the states and the desired trajectories.
    
   \begin{figure}[H]
   \begin{subfigure}[t]{0.32\linewidth} 
       \includegraphics[trim={1.4cm 7.33cm 18.2cm 9cm},clip, width=1.17\linewidth]{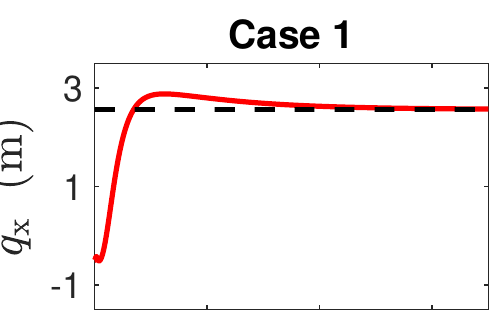}
   \end{subfigure}
    \hspace{0.25cm}
   \begin{subfigure}[t]{0.32\linewidth} 
       \includegraphics[trim={6.55cm 10.5cm 8.25cm 12cm},clip, width=0.95\linewidth]{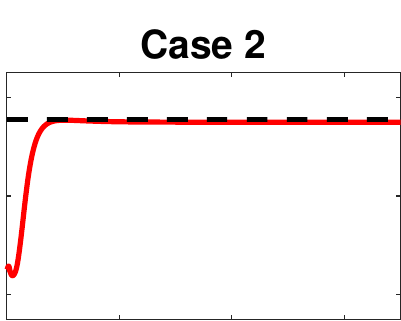}
   \end{subfigure}
   \hspace{-0.37cm}
   \begin{subfigure}[t]{0.32\linewidth} 
       \includegraphics[trim={13.3cm 10.475cm 1.4cm 12.5cm},clip, width=0.96\linewidth]{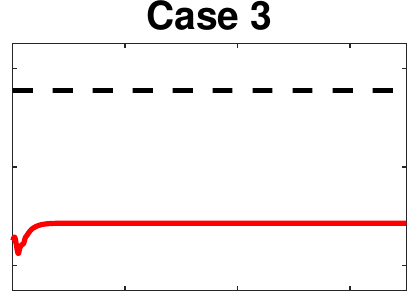}
   \end{subfigure}\\
   \begin{subfigure}[t]{0.32\linewidth} 
       \includegraphics[trim={1.4cm 7.33cm 18.2cm 10cm},clip, width=1.17\linewidth]{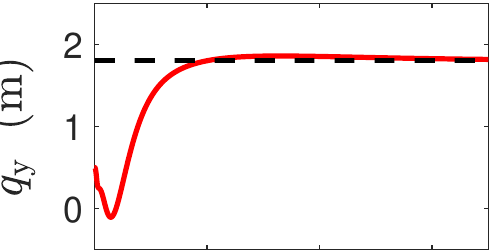}
   \end{subfigure}
   \hspace{0.25cm}
   \begin{subfigure}[t]{0.32\linewidth} 
       \includegraphics[trim={6.55cm 10.5cm 8.25cm 13.2cm},clip, width=0.95\linewidth]{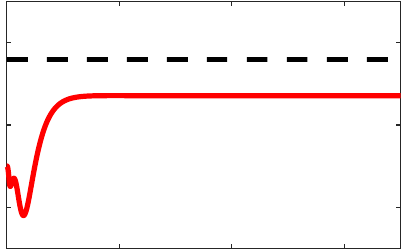}
   \end{subfigure}
   \hspace{-0.37cm}
   \begin{subfigure}[t]{0.32\linewidth} 
       \includegraphics[trim={13.3cm 10.5cm 1.4cm 13.2cm},clip, width=0.96\linewidth]{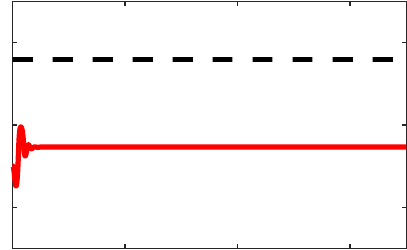}
   \end{subfigure}\\   
   \begin{subfigure}[t]{0.32\linewidth} 
       \includegraphics[trim={1.4cm 7.33cm 18.2cm 10cm},clip, width=1.17\linewidth]{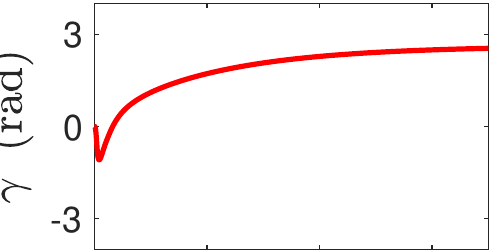}
   \end{subfigure}
   \hspace{0.25cm}
   \begin{subfigure}[t]{0.32\linewidth} 
       \includegraphics[trim={6.55cm 10.5cm 8.25cm 13.2cm},clip, width=0.95\linewidth]{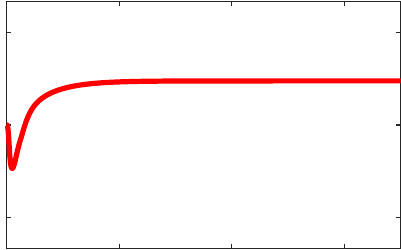}
   \end{subfigure}
   \hspace{-0.37cm}
   \begin{subfigure}[t]{0.32\linewidth} 
       \includegraphics[trim={13.3cm 10.5cm 1.4cm 13.2cm},clip, width=0.96\linewidth]{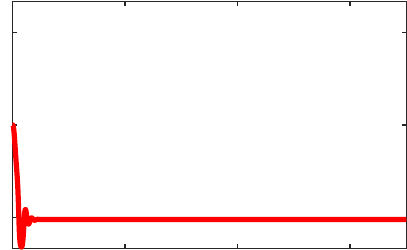}
   \end{subfigure}\\
   \begin{subfigure}[t]{0.32\linewidth} 
       \includegraphics[trim={1.4cm 7.34cm 18.2cm 10cm},clip, width=1.17\linewidth]{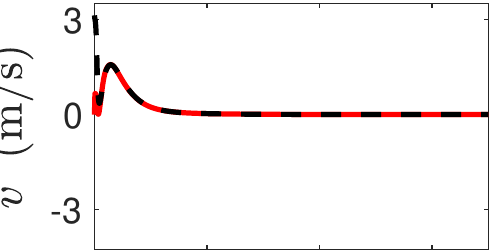}
   \end{subfigure}
   \hspace{0.25cm}
   \begin{subfigure}[t]{0.32\linewidth} 
       \includegraphics[trim={6.55cm 10.5cm 8.25cm 13.2cm},clip, width=0.95\linewidth]{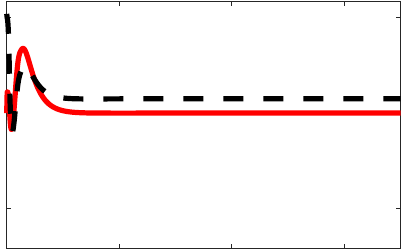}
   \end{subfigure}
   \hspace{-0.37cm}
   \begin{subfigure}[t]{0.32\linewidth} 
       \includegraphics[trim={13.3cm 10.5cm 1.4cm 13.2cm},clip, width=0.96\linewidth]{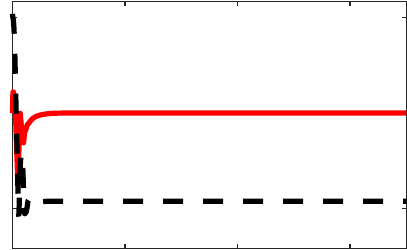}
   \end{subfigure}\\
   \begin{subfigure}[t]{0.32\linewidth} 
       \includegraphics[trim={1.38cm 5.62cm 18.2cm 10cm},clip, width=1.17\linewidth]{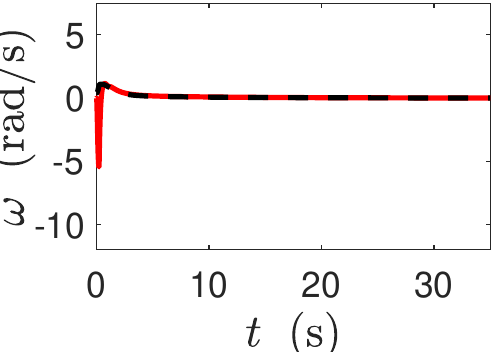}
   \end{subfigure}
   \hspace{0.25cm}
   \begin{subfigure}[t]{0.32\linewidth} 
       \includegraphics[trim={6.55cm 8.78cm 8.25cm 13.2cm},clip, width=0.95\linewidth]{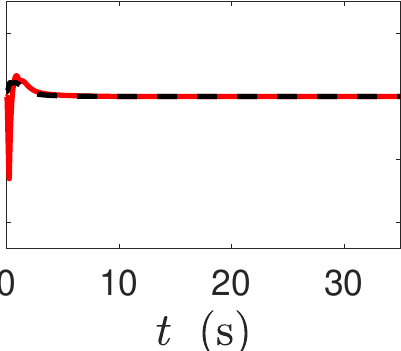}
   \end{subfigure}
   \hspace{-0.37cm}
   \begin{subfigure}[t]{0.32\linewidth} 
       \includegraphics[trim={13.3cm 8.75cm 1.4cm 13.2cm},clip, width=0.96\linewidth]{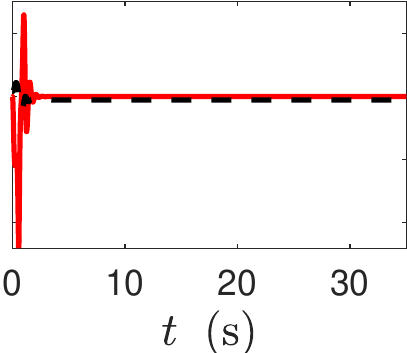}
   \end{subfigure}
   \caption{$q_{\rm{x}}$, $q_{\rm{y}}$, $\gamma$, $v$ and $\omega$ for Cases 1, 2 and 3. Note that $q_{\rm{dx}}$, $q_{\rm{dy}}$, $v_{\rm{d}}$ and $\omega_{\rm{d}}$ are shown with dashed lines.}
   \label{fig:GR:states}
   \end{figure}

    \vspace{-4mm}
   \begin{figure}[H]
   \centering
   \begin{subfigure}[t]{0.32\linewidth} 
       \includegraphics[trim={1.4cm 7.33cm 18.2cm 9.25cm},clip, width=1.17\linewidth]{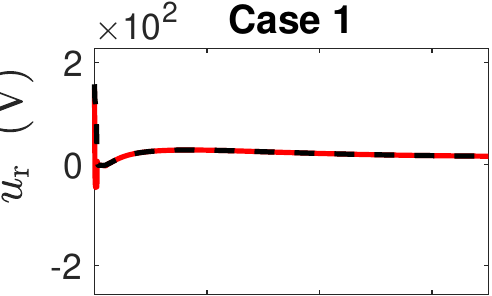}
   \end{subfigure}
   \hspace{0.25cm}
   \begin{subfigure}[t]{0.32\linewidth} 
       \includegraphics[trim={6.55cm 10.5cm 8.25cm 12.5cm},clip, width=0.95\linewidth]{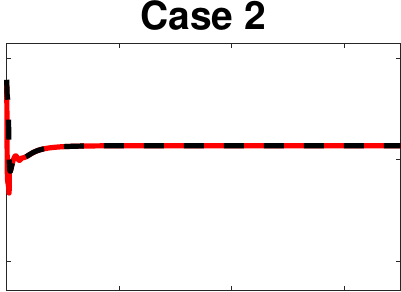}
   \end{subfigure}
   \hspace{-0.37cm}
   \begin{subfigure}[t]{0.32\linewidth} 
       \includegraphics[trim={13.35cm 10.5cm 1.4cm 12.5cm},clip, width=0.96\linewidth]{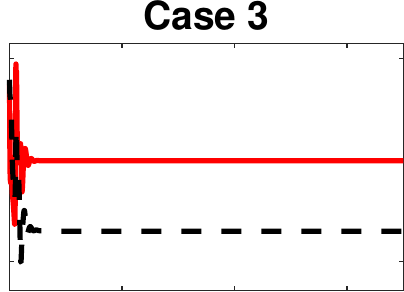}
  \end{subfigure}\\
   \begin{subfigure}[t]{0.32\linewidth} 
       \includegraphics[trim={1.4cm 5.68cm 18.2cm 10cm},clip, width=1.17\linewidth]{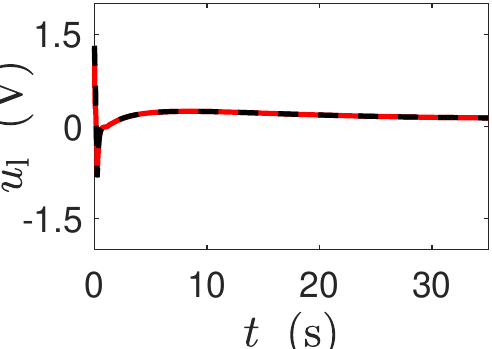}
   \end{subfigure}
   \hspace{0.25cm}
   \begin{subfigure}[t]{0.32\linewidth} 
       \includegraphics[trim={6.55cm 8.83cm 8.25cm 13.2cm},clip, width=0.95\linewidth]{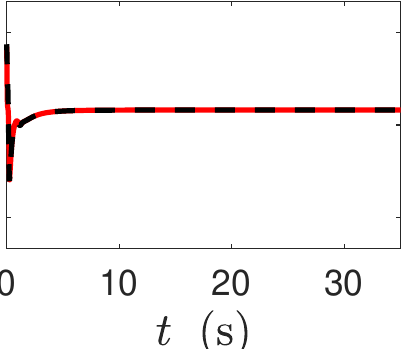}
   \end{subfigure}
   \hspace{-0.37cm}
   \begin{subfigure}[t]{0.32\linewidth} 
       \includegraphics[trim={13.35cm 8.8cm 1.4cm 13.2cm},clip, width=0.95\linewidth]{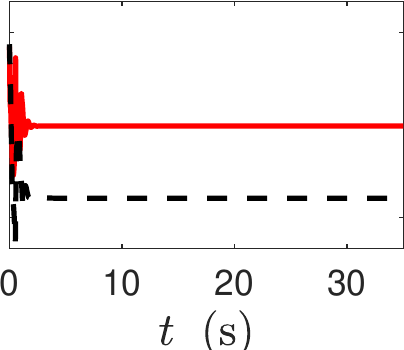}
   \end{subfigure}\\
   \caption{$u_{\rm{r}}$ and $u_{\rm{l}}$ for Cases 1, 2 and 3. Note that $u_{\rm{dr}}$ and $u_{\rm{dl}}$ are shown with dashed lines.}
   \label{fig:GR:control}
   \end{figure}

    \vspace{-4mm}
    \begin{figure}[H]
    \centering
    \begin{subfigure}[t]{0.32\linewidth} 
        \includegraphics[trim={1.4cm 7.33cm 18.2cm 9.25cm},clip, width=1.17\linewidth]{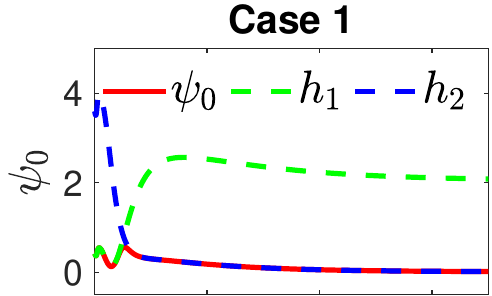}
    \end{subfigure}
    \hspace{0.25cm}
    \begin{subfigure}[t]{0.32\linewidth} 
        \includegraphics[trim={6.55cm 10.5cm 8.25cm 12.5cm},clip, width=0.95\linewidth]{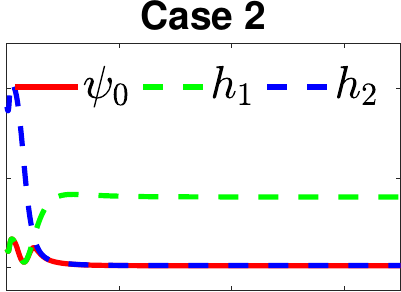}
    \end{subfigure}
    \hspace{-0.37cm}
    \begin{subfigure}[t]{0.32\linewidth} 
        \includegraphics[trim={13.35cm 10.5cm 1.4cm 12.5cm},clip, width=0.96\linewidth]{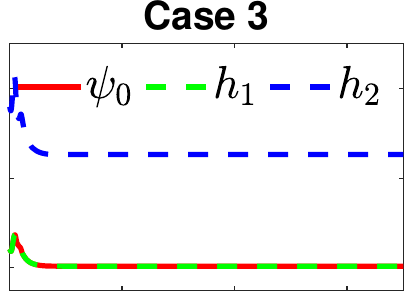}
    \end{subfigure}\\
    \begin{subfigure}[t]{0.32\linewidth} 
        \includegraphics[trim={1.4cm 7.33cm 18.2cm 10cm},clip, width=1.17\linewidth]{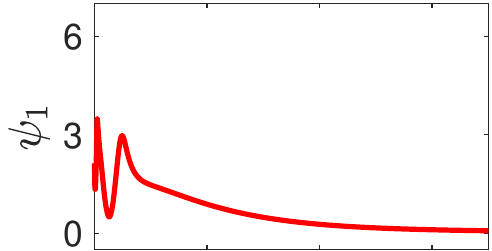}
    \end{subfigure}
    \hspace{0.25cm}
    \begin{subfigure}[t]{0.32\linewidth} 
        \includegraphics[trim={6.55cm 10.5cm 8.25cm 13.2cm},clip, width=0.95\linewidth]{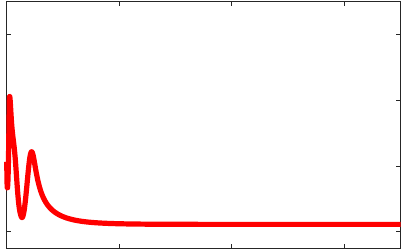}
    \end{subfigure}
    \hspace{-0.37cm}
    \begin{subfigure}[t]{0.32\linewidth} 
        \includegraphics[trim={13.35cm 10.5cm 1.4cm 13.2cm},clip, width=0.95\linewidth]{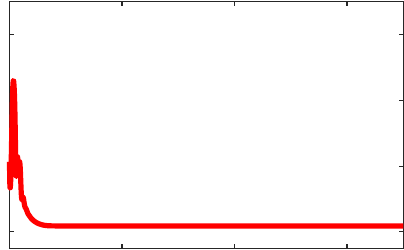}
    \end{subfigure}\\
    \begin{subfigure}[t]{0.32\linewidth} 
        \includegraphics[trim={1.4cm 5.63cm 18.2cm 10cm},clip, width=1.17\linewidth]{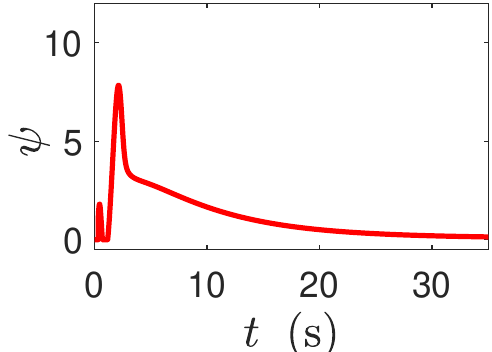}
    \end{subfigure}
    \hspace{0.25cm}
    \begin{subfigure}[t]{0.32\linewidth} 
        \includegraphics[trim={6.55cm 8.8cm 8.25cm 13.2cm},clip, width=0.95\linewidth]{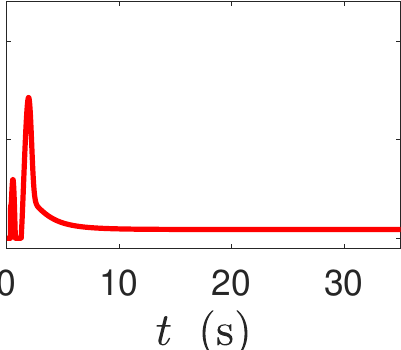}
    \end{subfigure}
    \hspace{-0.37cm}
    \begin{subfigure}[t]{0.32\linewidth} 
        \includegraphics[trim={13.35cm 8.75cm 1.4cm 13.2cm},clip, width=0.95\linewidth]{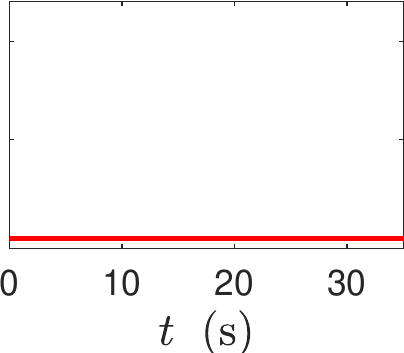}
    \end{subfigure}
    \caption{$\psi_{0}$, $\psi_{1}$ and $\psi$ for Cases 1, 2 and 3. }
    \label{fig:GR:safety}
    \end{figure}
    \vspace{-2mm}

    Figure \ref{fig:GR:control} shows that only in Case 3 do the control inputs $u_{\rm{r}}$ and $u_{\rm{l}}$ not follow the desired controls $u_{\rm{dr}}$ and $u_{\rm{dl}}$.
    
    \begin{figure}[H]
    \centering
    \begin{subfigure}[t]{0.32\linewidth} 
        \includegraphics[trim={1cm 5.1cm 18.6cm 5.5cm},clip, width=1.19\linewidth]{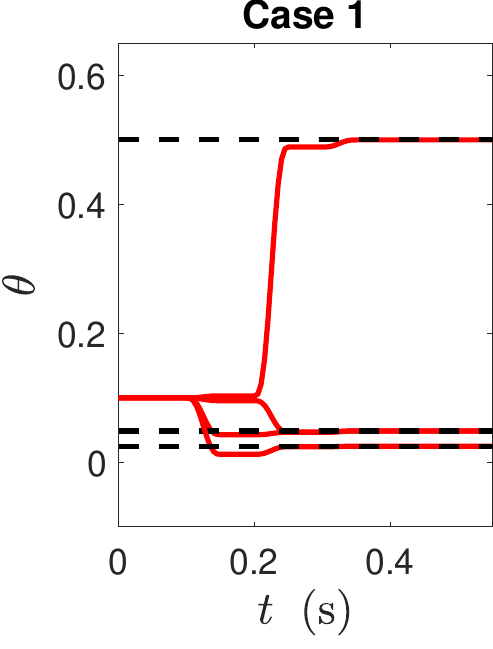}
    \end{subfigure}
    \hspace{0.33cm}
    \begin{subfigure}[t]{0.32\linewidth} 
        \includegraphics[trim={8cm 8.25cm 7cm 8.5cm},clip, width=0.94\linewidth]{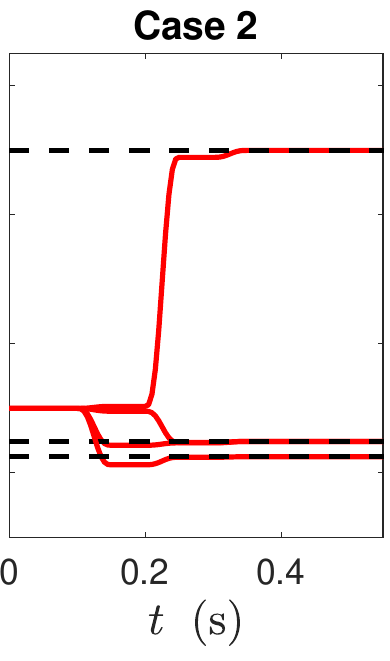}
    \end{subfigure}
    \hspace{-0.42cm}
    \begin{subfigure}[t]{0.32\linewidth} 
        \includegraphics[trim={19.5cm 5.1cm 1.95cm 5.5cm},clip, width=0.93\linewidth]{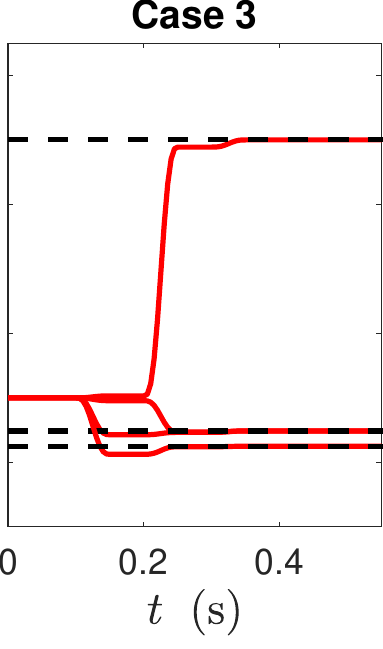}
    \end{subfigure}
        \caption{$\theta(t)$ for Cases 1, 2 and 3. Note that $\theta_{*}$ is shown with dashed lines.}
    \label{fig:GR:est}
    \end{figure}
    \vspace{-2mm}
    
    Figure \ref{fig:GR:safety} shows that $\psi_{0}$, $\psi_{1}$ and $\psi$ are nonnegative, which implies that for all three cases and all $t \geq 0$, $x(t) \in C_{0}$. In Case 3, for all $t\geq 0$, $\psi(x(t))=0$, demonstrating how the conservative safety constraint degrades performance, even as the desired control is adaptively improved.
    
    Figure \ref{fig:GR:est} and \ref{fig:GR:est_error} show that, for all three cases, $\lim_{t\rightarrow \infty}\theta(t)=\theta_{*}$ and $\lim_{t \rightarrow\infty}\nu (t) = 0$.

    \vspace{-2mm}
    \begin{figure}[H]
    \centering
    \begin{subfigure}[t]{0.32\linewidth} 
        \includegraphics[trim={1.4cm 5.6cm 18.2cm 9cm},clip, width=1.17\linewidth]{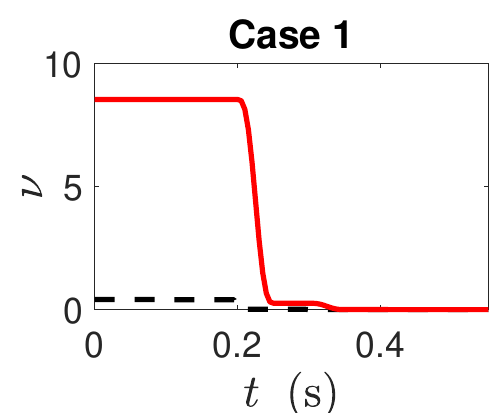}
    \end{subfigure}
    \hspace{0.22cm}
    \begin{subfigure}[t]{0.32\linewidth} 
        \includegraphics[trim={6.55cm 8.8cm 8.25cm 12cm},clip, width=0.96\linewidth]{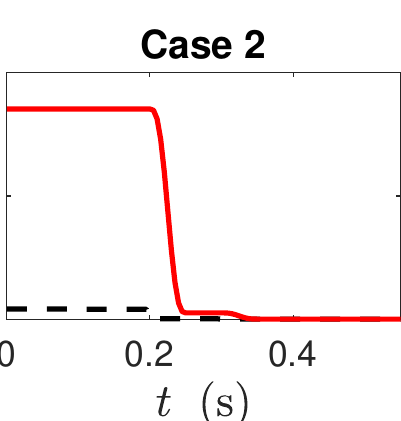}
    \end{subfigure}
    \hspace{-0.34cm}
    \begin{subfigure}[t]{0.32\linewidth} 
        \includegraphics[trim={16.5cm 5.6cm 4.55cm 9cm},clip, width=0.97\linewidth]{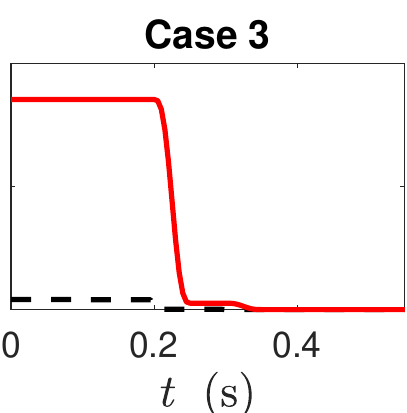}
    \end{subfigure}
    \caption{$\nu$ for Cases 1, 2 and 3. Note that $||\tilde{\theta}_{k}||$ is shown with dashed lines.}
    \label{fig:GR:est_error}
    \end{figure}

    \bibliographystyle{ieeetr}
    \bibliography{References}

 \end{document}